\documentclass[12pt, draftcls, onecolumn]{IEEEtran}%
\usepackage{amssymb}
\usepackage{amsmath}
\usepackage{graphicx}
\usepackage{amsfonts}%
\newtheorem{thm}{Theorem}[section]

\newtheorem{prop}[thm]{Proposition}

\newtheorem{rem}{Remark}[section]
\begin{document}
\vspace{-.5in}
\title{Compound Multiple Access Channels with Partial Cooperation \thanks{This work
was supported by NSF under grants CNS-06-26611 and CNS-06-25637, the DARPA
ITMANET program under grant 1105741-1-TFIND and the ARO under MURI award
W911NF-05-1-0246. The work of S. Shamai has been supported by the Israel
Science Foundation and the European Commission in the framework of the FP7
Network of Excellence in Wireless COMmunications NEWCOM++.}}
\author{O. Simeone$^{1}$, D. G\"{u}nd\"{u}z$^{2,3}$, H. V. Poor$^{2}$, A.
Goldsmith$^{3}$ and S. Shamai (Shitz)$^{4}$\\
$^{1}$CWCSPR, New Jersey Institute of Technology, Newark, NJ 07102, USA\\
\and $^{2}$Dept. of Electrical Engineering, Princeton Univ., Princeton, NJ 08544, USA\\
\and $^{3}$Dept. of Electrical Engineering, Stanford Univ., Stanford, CA 94305, USA\\
\and $^{4}$Dept. of Electrical Engineering, Technion, Haifa, 32000, Israel\\
\and \texttt{{\small \{osvaldo.simeone@njit.edu, dgunduz@princeton.edu,
poor@princeton.edu, andrea@stanford.edu, sshlomo@ee.technion.ac.il\}.}}}
\maketitle
\vspace{-.8in} 
\begin{abstract}
A two-user discrete memoryless compound multiple access channel with a common
message and conferencing decoders is considered. The capacity region is
characterized in the special cases of physically degraded channels and
unidirectional cooperation, and achievable rate regions are provided for the
general case. The results are then extended to the corresponding Gaussian
model. In the Gaussian setup, the provided achievable rates are shown to lie
within some constant number of bits from the boundary of the capacity region
in several special cases. An alternative model, in which the encoders are
connected by conferencing links rather than having a common message, is
studied as well, and the capacity region for this model is also determined for
the cases of physically degraded channels and unidirectional cooperation.
Numerical results are also provided to obtain insights about the potential
gains of conferencing at the decoders and encoders.
\end{abstract}

\thispagestyle{empty} \pagestyle{empty}


\section{Introduction}

In today's complex communication networks there are often multiple ``signal
paths'' to utilize in delivering data between a given transmitter and
receiver. Such signal paths may take the form of (generalized) feedback from
the channel to the transmitters or additional (orthogonal) communication links
between either the transmitters or the receivers. The first case corresponds
to scenarios in which the additional signal paths share the spectral resources
with the direct transmitter-receiver links (in-band signalling), while the
second case refers to scenarios in which orthogonal spectral resources are
available at the transmit and/or the receive side (out-of-band signalling).

In this work, we focus on the latter case discussed above and model the
out-of-band signal paths as finite-capacity directed links. This framework is
typically referred to as ``conferencing'' (or ``partial cooperation'') in the
literature to emphasize the possibly interactive nature of communication on
such links. Conferencing encoders in a two-user multiple access channel (MAC)
have been investigated in \cite{Willems} \cite{Lapidoth}\footnote{It is noted
that a MAC with conferencing encoders can be seen as a special case of a
MAC\ with generalized feedback.} and in \cite{Maric Yates Kramer 2007} for a
two-user interference channel. These works show that conferencing encoders can
create dependence between the transmitted signals by coordinating the
transmission via the out-of-band links, thus mimicking multiantenna
transmitters. Conferencing decoders have been studied in \cite{Ng Maric et al
2006} for a relay channel and in \cite{Dabora Servetto 2006} - \cite{Lasaulce}
for a broadcast channel. Such decoders can use the out-of-band links to
exchange messages about the received signals so as to mimic a multiantenna
receiver (see also \cite{ETH master}).

This work extends the state of the art described above by considering the
compound MAC with conferencing decoders and a common message (see Fig.
\ref{model}) and then with both conferencing encoders and decoders (see Fig.
\ref{model1}). These models generalize the setup of a single-message broadcast
(multicast) channel with two conferencing decoders studied in \cite{Dabora
Servetto 2006}\footnote{Reference \cite{Dabora Servetto 2006} also considers a
broadcast channel with private messages to the two users.} - \cite{Lasaulce},
in that there are two transmitters that want to broadcast their messages to
the conferencing receivers.\ Moreover, the transmitters can have a common
message (Fig. \ref{model}) or be connected by conferencing links (Fig.
\ref{model1}). The model also generalizes the compound MAC with common
messages studied, among other models, in \cite{Maric Yates Kramer 2007}, by
allowing conferencing among the decoders. The main contributions of this work
are summarized as follows:

\begin{itemize}
\item The capacity region is derived for the two-user discrete-memoryless
compound MAC with a common message and conferencing decoders in the special
cases of physically degraded channels and unidirectional cooperation (Sec.
\ref{sec_capacity_1});

\item Achievable rate regions are given for the general model of Fig.
\ref{model} (Sec. \ref{sec_achievability});

\item Extension to the corresponding Gaussian case is provided, establishing
the capacity region with unidirectional cooperation and deriving general
achievable rates. Such achievable rates are also shown to be within some
constant number of bits of the capacity region in several special cases (Sec.
\ref{sec_Gaussian});

\item The capacity region is determined for the compound MAC\ with both
conferencing encoders and decoders as in Fig. \ref{model1} in the special
cases of physically degraded channels and unidirectional cooperation (Sec.
\ref{sec_enc dec}).
\end{itemize}

Finally, numerical results are also provided to obtain further insight into
the main conclusions.

\section{System Model and Main Definitions\label{sec_model}}

We start by considering the model in Fig. \ref{model}, which is a
discrete-memoryless compound MAC with conferencing decoders and common
information (here, for short, we will refer to this channel as the
CM\ channel). The CM\ channel is characterized by $(\mathcal{X}_{1},$
$\mathcal{X}_{2},$ $p^{\ast}(y_{1},y_{2}|x_{1},x_{2}),$ $\mathcal{Y}_{1},$
$\mathcal{Y}_{2})$ with input alphabets $\mathcal{X}_{1},\mathcal{X}_{2}$ and
output alphabets $\mathcal{Y}_{1},\mathcal{Y}_{2}.$ The $i$th encoder,
$i=1,2$, is interested in sending a private message $W_{i}\in\mathcal{W}%
_{i}=\{1,2,...,2^{nR_{i}}\}$ of rate $R_{i}$ [bits/channel use] to both
receivers and, in addition, there is a common message $W_{0}\in\mathcal{W}%
_{0}=\{1,2,...,2^{nR_{0}}\}$ of rate $R_{0}$ to be delivered by both encoders
to both receivers. The channel is memoryless and time-invariant in that the
conditional distribution of the output symbols at any time $j=1,2,...,n$
satisfies%
\begin{equation}
p(y_{1,j},y_{2,j}|x_{1}^{n},x_{2}^{n},y_{1}^{j-1},y_{2}^{j-1},\bar{w}%
)=p^{\ast}(y_{1,j},y_{2,j}|x_{1,j},x_{2,j})
\end{equation}
with $\bar{w}=(w_{0},w_{1},w_{2})\in\mathcal{W}_{0}\times\mathcal{W}_{1}%
\times\mathcal{W}_{2}$ being a given triplet of messages. Notation-wise, we
employ standard conventions (see, e.g., \cite{Cover}), where the probability
distributions are defined by the arguments, upper-case letters represent
random variables and the corresponding lower-case letters represent
realizations of the random variables. The superscripts identify the number of
samples to be included in a given vector, e.g., $y_{1}^{j-1}=[y_{1,1}\cdots
y_{1,j-1}].$ It is finally noted that the channel defines the conditional
marginals $p(y_{1}|x_{1},x_{2})=\sum_{y_{2}\in\mathcal{Y}_{2}}p^{\ast}%
(y_{1},y_{2}|x_{1},x_{2})$ and similarly for $p(y_{2}|x_{1},x_{2}).$ Further
definitions are in order.

\textit{Definition 1}: A $((2^{nR_{0}},2^{nR_{1}},2^{nR_{2}}),n,K)$ code for
the CM\ channel consists of two encoding functions ($i=1,2$)%
\begin{equation}
f_{i}\text{: }\mathcal{W}_{0}\times\mathcal{W}_{i}\rightarrow\mathcal{X}%
_{i}^{n},\text{ }%
\end{equation}
a set of $2K$ ``conferencing'' functions and corresponding output alphabets
$\mathcal{V}_{i,k}$ ($k=1,2,...,K$):
\begin{subequations}
\label{conferencing functions}%
\begin{align}
g_{1,k}  &  \text{:}\text{ }\mathcal{Y}_{1}^{n}\times\mathcal{V}_{2,1}%
\times\cdots\times\mathcal{V}_{2,k-1}\rightarrow\mathcal{V}_{1,k}\\
g_{2,k}  &  \text{:}\text{ }\mathcal{Y}_{2}^{n}\times\mathcal{V}_{1,1}%
\times\cdots\times\mathcal{V}_{1,k-1}\rightarrow\mathcal{V}_{2,k},
\end{align}
and decoding functions:
\end{subequations}
\begin{subequations}
\label{decoding}%
\begin{align}
h_{1}  &  \text{:}\text{ }\mathcal{Y}_{1}^{n}\times\mathcal{V}_{2,1}%
\times\cdots\times\mathcal{V}_{2,K}\rightarrow\mathcal{W}_{0}\times
\mathcal{W}_{1}\\
h_{2}  &  \text{:}\text{ }\mathcal{Y}_{2}^{n}\times\mathcal{V}_{1,1}%
\times\cdots\times\mathcal{V}_{1,K}\rightarrow\mathcal{W}_{0}\times
\mathcal{W}_{2}.
\end{align}

Notice that the conferencing functions (\ref{conferencing functions})
prescribe $K$\ conferencing rounds between the decoders that start as soon as
the two decoders receive the entire block of $n$ output symbols $y_{1}^{n}$
and $y_{2}^{n}.$ Each conference round, say the $k$th, corresponds to a
simultaneous and bidirectional exchange of messages between the two decoders
taken from the alphabets $\mathcal{V}_{1,k}$ and $\mathcal{V}_{2,k}$,
similarly to \cite{Willems}, \cite{Liang Kramer}. It is noted that other works
have used slightly different definitions of conferencing rounds \cite{Draper
et al 2003}, \cite{Steiner}. After the $K$ conferencing rounds, the receivers
perform decoding with functions (\ref{decoding}) by capitalizing on the
exchanged conferencing messages. Due to the orthogonality between the main
channel and the conferencing links, the transmission from the users on one
channel and conferencing/ decoding on the other can take place simultaneously.%

\begin{figure}
\begin{center}
\includegraphics[width=4.0421in]%
{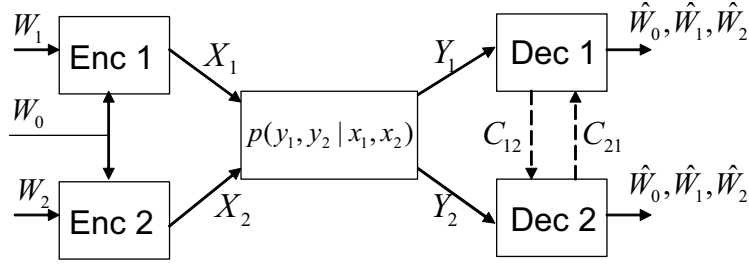}%
\caption{A discrete-memoryless compound MAC channel with conferencing decoders
and common information (for short, CM).}%
\label{model}%
\end{center}
\end{figure}

\textit{Definition 2}: A rate triplet ($R_{0},R_{1},R_{2}$) is said to be
\textit{achievable} for the CM\ channel with decoders connected by
conferencing links with capacities $(C_{12},C_{21})$ (see Fig. \ref{model}) if
for any $\varepsilon>0$ there exists, for all $n$ sufficiently large, a
$((2^{nR_{0}},2^{nR_{1}},2^{nR_{2}}),n,K)$ code with any $K\geq0$ such that
the probability of error satisfies%

\end{subequations}
\begin{equation}
P_{e}\triangleq\frac{1}{2^{n(R_{0}+R_{1}+R_{2})}}\sum_{\bar{w}\in
\mathcal{W}_{0}\times\mathcal{W}_{1}\times\mathcal{W}_{2}}\Pr\left[
\begin{array}
[c]{c}%
\{h_{1}(Y_{1}^{n},V_{2}^{k})\neq\bar{w}\}\cup\\
\{h_{2}(Y_{2}^{n},V_{1}^{k})\neq\bar{w}\}|\bar{w}\text{ sent}%
\end{array}
\right]  \leq\varepsilon\label{Pe1}%
\end{equation}
and the conferencing alphabets are such that
\begin{equation}%
{\displaystyle\sum\limits_{k=1}^{K}}
|\mathcal{V}_{1,k}|\leq nC_{12}\text{ and }%
{\displaystyle\sum\limits_{k=1}^{K}}
|\mathcal{V}_{2,k}|\leq nC_{21}. \label{C constraint}%
\end{equation}
The \textit{capacity region} $\mathcal{C}_{CM}(C_{12},C_{21})$ is the closure
of the set of all achievable rates ($R_{0},R_{1},R_{2}$).

\section{Preliminaries and an Outer bound}

Similarly to \cite{Maric Yates Kramer 2007}, it is useful to define the rate
region $\mathcal{R}_{MAC,i}(p(u),p(x_{1}|u),p(x_{2}|u))$ for the MAC\ seen at
the $i$th receiver ($i=1,2$) as the set of rates
\begin{subequations}
\begin{align}
\mathcal{R}_{MAC,i}\big(p(u),p(x_{1}|u),p(x_{2}|u)\big) = \bigg\{(R_{0}%
,R_{1},R_{2})  &  \text{:}R_{0}\geq0,\text{ }R_{1}\geq0,\text{ } R_{2}%
\geq0,\nonumber\\
R_{1}  &  \leq I(X_{1};Y_{i}|X_{2}U),\text{ }\\
R_{2}  &  \leq I(X_{2};Y_{i}|X_{1}U),\\
R_{1}+R_{2}  &  \leq I(X_{1}X_{2};Y_{i}|U),\\
R_{0}+R_{1}+R_{2}  &  \leq I(X_{1}X_{2};Y_{i})\bigg\},
\end{align}
where the joint distributions of the involved variables is given by
\end{subequations}
\[
p(u)p(x_{1}|u)p(x_{2}|u)p(y_{i}|x_{1},x_{2}).
\]
If $C_{12}=C_{21}=0,$ the capacity region $\mathcal{C}_{CM}(0,0)$ is given by
\cite{Maric Yates Kramer 2007}:
\begin{subequations}
\label{Ccm00}%
\begin{align}
\mathcal{C}_{CM}(0,0)  &  =%
{\displaystyle\bigcup}
\left\{  \underset{i=1,2}{%
{\displaystyle\bigcap}
}\mathcal{R}_{MAC,i}(p(u),p(x_{1}|u),p(x_{2}|u))\right\} \\
&  =%
{\displaystyle\bigcup}
\bigg\{(R_{0},R_{1},R_{2})\text{:}\text{ }R_{0}\geq0,\text{ }R_{1}\geq0,\text{
}R_{2}\geq0,\\
R_{1}  &  \leq\min\{I(X_{1};Y_{1}|X_{2}U),\text{ }I(X_{1};Y_{2}|X_{2}U)\},\\
R_{2}  &  \leq\min\{I(X_{2};Y_{1}|X_{1}U),\text{ }I(X_{2};Y_{2}|X_{1}U)\},\\
R_{1}+R_{2}  &  \leq\min\{I(X_{1}X_{2};Y_{1}|U),\text{ }I(X_{1}X_{2}%
;Y_{2}|U)\},\\
R_{0}+R_{1}+R_{2}  &  \leq\min\{I(X_{1}X_{2};Y_{1}),\text{ }I(X_{1}X_{2}%
;Y_{2})\}\bigg\},
\end{align}
where the union is taken over all joint distributions of the form
\end{subequations}
\[
p(u)p(x_{1}|u)p(x_{2}|u)p^{\ast}(y_{1},y_{2}|x_{1},x_{2}).
\]
It is remarked that no convex hull operation is necessary in (\ref{Ccm00}) as
the region $\mathcal{C}_{CM}(0,0)$ is convex \cite{Maric Yates Kramer 2007}
(see also \cite{Willems}, Appendix \ref{app:1}).

We now derive an outer bound on the capacity region $\mathcal{C}_{CM}%
(C_{12},C_{21})$. To this end, it is useful to define the capacity region
achievable when the two receivers are allowed to fully cooperate (FC), thus
forming a two-antenna receiver. In this case, we have
\begin{subequations}
\begin{align}
\mathcal{R}_{MAC,FC}\big(p(u),p(x_{1}|u),p(x_{2}|u)\big) = \bigg\{(R_{0}%
,R_{1},R_{2})  &  \text{:}R_{0}\geq0,\text{ }R_{1}\geq0,\text{ }R_{2}\geq0,\\
R_{1}  &  \leq I(X_{1};Y_{1}Y_{2}|X_{2}U),\text{ }\\
R_{2}  &  \leq I(X_{2};Y_{1}Y_{2}|X_{1}U),\\
R_{1}+R_{2}  &  \leq I(X_{1}X_{2};Y_{1}Y_{2}|U),\\
R_{0}+R_{1}+R_{2}  &  \leq I(X_{1}X_{2};Y_{1}Y_{2})\bigg\},
\end{align}
where the joint distributions of the involved variables is given by
\end{subequations}
\begin{equation}
p(u)p(x_{1}|u)p(x_{2}|u)p^{\ast}(y_{1},y_{2}|x_{1},x_{2}). \label{pupxpy}%
\end{equation}

\begin{prop}
\label{prop1} We have $\mathcal{C}_{CM}(C_{12},C_{21})\subseteq\mathcal{C}%
_{CM-out}(C_{12},C_{21})$ where (dropping the dependence on $p(u),p(x_{1}%
|u),p(x_{2}|u)$ to simplify the notation)
\begin{subequations}
\label{outer bound}%
\begin{align}
\mathcal{C}_{CM-out}(C_{12},C_{21})  &  =%
{\displaystyle\bigcup}
\big\{ (\mathcal{R}_{MAC,1}+C_{12})\cap(\mathcal{R}_{MAC,2}+C_{21}%
)\cap(\mathcal{R}_{MAC,FC})\big\}\\
&  =%
{\displaystyle\bigcup}
\bigg\{(R_{0},R_{1},R_{2})\text{:}\text{ }R_{0}\geq0,\text{ }R_{1}\geq0,\text{
}R_{2}\geq0,\\
&  R_{1} \leq\min\{I(X_{1};Y_{1}|X_{2}U)+C_{21},\text{ }I(X_{1};Y_{2}%
|X_{2}U)+C_{12},\text{ }\nonumber\\
&  ~~~~~~~ I(X_{1};Y_{1}Y_{2}|X_{2}U)\},\\
&  R_{2} \leq\min\{I(X_{2};Y_{1}|X_{1}U)+C_{21},\text{ }I(X_{2};Y_{2}%
|X_{1}U)+C_{12},\text{ }\nonumber\\
&  ~~~~~~~ I(X_{2};Y_{1}Y_{2}|X_{1}U)\},\\
&  R_{1}+R_{2} \leq\min\{I(X_{1}X_{2};Y_{1}|U)+C_{21},\text{ }I(X_{1}%
X_{2};Y_{2}|U)+C_{12},\text{ }\label{last-conv}\\
&  ~~~~~~~ I(X_{1}X_{2};Y_{1}Y_{2}|U)\},\\
&  R_{0}+R_{1}+R_{2} \leq\min\{I(X_{1}X_{2};Y_{1})+C_{21},\text{ }I(X_{1}%
X_{2};Y_{2})+C_{12},\nonumber\\
&  ~~~~~~~ I(X_{1}X_{2};Y_{1}Y_{2})\}\bigg\},
\end{align}
in which the union is taken over all the joint distributions that factorize as
(\ref{pupxpy}).
\end{subequations}
\end{prop}

Similarly to (\ref{Ccm00}), region (\ref{outer bound}) can be proven to be
convex following \cite{Willems}.

\begin{proof}
See Appendix \ref{app:1}.
\end{proof}

\section{Capacity Region with Physically Degraded Channels and Unidirectional
Cooperation}

\label{sec_capacity_1}

The next proposition establishes the capacity region $\mathcal{C}%
_{CM-DEG}(C_{12},C_{21})$ in the case of physically degraded outputs.

\begin{prop}
\label{prop2} If the CM\ channel is physically degraded in the sense that
$(X_{1}X_{2})-Y_{1}-Y_{2}$ forms a Markov chain, then the capacity region is
obtained as
\begin{subequations}
\begin{align}
\mathcal{C}_{CM-DEG}(C_{12},C_{21})  &  =\mathcal{C}_{CM-out}(C_{12},0)=\\
&  =%
{\displaystyle\bigcup}
\{(R_{0},R_{1},R_{2})\text{:}\text{ }R_{0}\geq0,\text{ }R_{1}\geq0,\text{
}R_{2}\geq0,\\
R_{1}  &  \leq\min\{I(X_{1};Y_{1}|X_{2}U),\text{ }I(X_{1};Y_{2}|X_{2}%
U)+C_{12}\},\\
R_{2}  &  \leq\min\{I(X_{2};Y_{1}|X_{1}U),\text{ }I(X_{2};Y_{2}|X_{1}%
U)+C_{12}\},\\
R_{1}+R_{2}  &  \leq\min\{I(X_{1}X_{2};Y_{1}|U),\text{ }I(X_{1}X_{2}%
;Y_{2}|U)+C_{12}\},\\
R_{0}+R_{1}+R_{2}  &  \leq\min\{I(X_{1}X_{2};Y_{1}),\text{ }I(X_{1}X_{2}%
;Y_{2})+C_{12}\}\}.
\end{align}
Notice that here $p^{\ast}(y_{1}y_{2}|x_{1},x_{2})=p(y_{1}|x_{1},x_{2}%
)p(y_{2}|y_{1})$ due to degradedness.
\end{subequations}
\end{prop}

\begin{proof}
See Appendix \ref{app:2}.
\end{proof}

\begin{rem}
A symmetric result clearly holds for the physically degraded channel
$(X_{1}X_{2})-Y_{2}-Y_{1}.$
\end{rem}

Establishment of the capacity region is also possible in the special case
where only unidirectional cooperation is allowed, that is $C_{12}=0$ or
$C_{21}=0.$ This result is akin to \cite{Lasaulce} where a broadcast channel
with two receiver under unidirectional cooperation was considered.

\begin{prop}
\label{prop3} In the case of unidirectional cooperation ($C_{12}=0$ or
$C_{21}=0$), the capacity region of the CM channel is given by
\begin{align}
\mathcal{C}_{CM}(0,C_{21})  &  =\mathcal{C}_{CM-out}(0,C_{21})
\end{align}
or
\begin{align}
\mathcal{C}_{CM}(C_{12},0)  &  =\mathcal{C}_{CM-out}(C_{12},0).
\end{align}

\end{prop}

\begin{proof}
Achievability follows by using the same scheme as in the proof of Proposition
\ref{prop2}. The converse is immediate.
\end{proof}

\section{General Achievable Rates\label{sec_achievability}}

Achievable rates can be derived for the general CM\ channel, extending the
analysis of \cite{Dabora Servetto 2006} from the broadcast setting with one
transmitter to the CM channel. Notice that \cite{Dabora Servetto 2006} uses a
different definition for the operation over the conferencing channels but this
turns out to be immaterial for the achievable rates discussed below.

\begin{prop}
\label{prop4} The following region is achievable with one-round
conferencing, i.e., $K=1$:
\begin{subequations}
\label{OS}%
\begin{align}
\mathcal{R}_{OR}(C_{12},C_{21})  &  =%
{\displaystyle\bigcup}
\{(R_{0},R_{1},R_{2})\text{:}\text{ }R_{0}\geq0,\text{ }R_{1}\geq0,\text{
}R_{2}\geq0,\\
R_{1}  &  \leq\min\{I(X_{1};Y_{1}\hat{Y}_{2}|X_{2}U),\text{ }I(X_{1};Y_{2}%
\hat{Y}_{1}|X_{2}U)\},\\
R_{2}  &  \leq\min\{I(X_{2};Y_{1}\hat{Y}_{2}|X_{1}U),\text{ }I(X_{2};Y_{2}%
\hat{Y}_{1}|X_{1}U)\},\\
R_{1}+R_{2}  &  \leq\min\{I(X_{1}X_{2};Y_{1}\hat{Y}_{2}|U),\text{ }%
I(X_{1}X_{2};Y_{2}\hat{Y}_{1}|U)\}\\
R_{0}+R_{1}+R_{2}  &  \leq\min\{I(X_{1}X_{2};Y_{1}\hat{Y}_{2}),\text{ }%
I(X_{1}X_{2};Y_{2}\hat{Y}_{1})\}\}
\end{align}
subject to
\end{subequations}
\begin{subequations}
\label{cap cond}%
\begin{align}
C_{12}  &  \geq I(Y_{1};\hat{Y}_{1}|Y_{2})\\
C_{21}  &  \geq I(Y_{2};\hat{Y}_{2}|Y_{1})
\end{align}
with $|\widehat{\mathcal{Y}}_{i}|\leq|\mathcal{Y}_{i}|+1$, and the union is
taken over all the joint distributions that factorize as
\end{subequations}
\[
p(u)p(x_{1}|u)p(x_{2}|u)p^{\ast}(y_{1},y_{2}|x_{1},x_{2})p(\hat{y}_{1}%
|y_{1})p(\hat{y}_{2}|y_{2}).
\]

\end{prop}

\begin{proof}
(Sketch): The proof is similar to that of Theorem 3 in \cite{Dabora Servetto
2006} and is thus only sketched here. A one-step conference ($K=1$) is used.
Encoding and transmission are performed as for a MAC with common information
(see proof of Proposition \ref{prop2}). Each receiver compresses its received
signal using Wyner-Ziv compression exploiting the fact that the other receiver
has a correlated observation as well. The compression indices are exchanged
during the one conferencing round via symbols $\mathcal{V}_{1,1}$ and
$\mathcal{V}_{2,1}.$ Decoding is then carried out at each receiver using joint
typicality: For instance, receiver 1 looks for jointly typical sequences
$(u^{n}(w_{0}),x_{1}^{n}(w_{0},w_{1}),x_{2}^{n}(w_{0},w_{2}),y_{1}^{n},\hat
{y}_{2}^{n})$ with $w_{i}\in\mathcal{W}_{i},$ where $\hat{y}_{2}^{n}$ is the
compressed sequence received by the second decoder.
\end{proof}

The achievable strategy of Proposition \ref{prop4} is based on $K=1$ round of
conferencing. It is easy to construct examples where such a strategy fails to
achieve the outer bound (\ref{outer bound}) as discussed in the example below.

\textit{Example 1}. Consider a symmetric scenario with $R_{0}=0 $ and
equal private rates $R_{1}=R_{2}=R$ (i.e., $p^{\ast}(y_{1}%
,y_{2}|x_{1},x_{2})$ $=p^{\ast}(y_{2},y_{1}|x_{1},x_{2})$ $=p^{\ast}(y_{1}%
,y_{2}|x_{2},x_{1})$ $=p^{\ast}(y_{2},y_{1}|x_{2},x_{1})$). Fix $U$ to a constant without loss of
generality (given the absence of a common message) and the input distribution
to $p(x_{1})p(x_{2}).$ We are interested in finding the maximum achievable
equal rate $R_{1}=R_{2}=R.$ Assume that the conferencing capacities satisfy
$C_{12}=H(Y_{1}|Y_{2})=H(Y_{2}|Y_{1})$ and $1/2\cdot I(X_{1}X_{2};Y_{2}%
|Y_{1})\leq C_{21}<H(Y_{1}|Y_{2}).$ In this case, it can be seen that the
maximum equal rate is upper bounded as $R\leq R_{out}=1/2\cdot I(X_{1}%
X_{2};Y_{1}Y_{2})$ by the outer bound (\ref{outer bound})$,$ which corresponds
to the maximum equal rate of a system with full cooperation at the receiver
side$.$ This bound can be achieved if both receivers have access to both
outputs $Y_{1}$ and $Y_{2}.$ With the one-round strategy, since $C_{12}%
=H(Y_{1}|Y_{2})$ receiver 1 can provide $Y_{1}$ to receiver 2 via Slepian-Wolf
compression, but receiver 2 cannot do the same with receiver 1 since
$C_{21}<H(Y_{1}|Y_{2}).$ Therefore, rate $R_{out}$ cannot be achieved by this
strategy, which in fact attains equal rate $R_{OR}=1/2\cdot I(X_{1}X_{2}%
;Y_{1}\hat{Y}_{2})<R_{out}$ (recall (\ref{cap cond})).

We now consider a second strategy that generalizes the previous one and is
based on two rounds of conferencing $(K=2).$ As will be shown below, this
strategy is able to improve upon the one-round scheme, while still failing to
achieve the outer-bound (\ref{outer bound}) in the general case.

\begin{prop}
\label{prop5} The following rate region is achievable with two rounds of
conferencing, i..e., $K=2$:
\begin{equation}
\mathcal{R}_{TR}(C_{12},C_{21})=\mathrm{co}%
{\displaystyle\bigcup}
\{\mathcal{R}_{TR,12}\cup\mathcal{R}_{TR,21}\} \label{Rts}%
\end{equation}
where ``co'' indicates the convex hull operation, and we have
\begin{subequations}
\label{Rtr1}%
\begin{align}
\mathcal{R}_{TR,12}  &  =\{(R_{0},R_{1},R_{2})\text{:}\text{ }R_{0}%
\geq0,\text{ }R_{1}\geq0,\text{ }R_{2}\geq0,\\
R_{1}  &  \leq\min\{I(X_{1};Y_{1}|X_{2}U)+C_{21},\text{ }I(X_{1};Y_{2}\hat
{Y}_{1}|X_{2}U)\},\\
R_{2}  &  \leq\min\{I(X_{2};Y_{1}|X_{1}U)+C_{21},\text{ }I(X_{2};Y_{2}\hat
{Y}_{1}|X_{1}U)\},\\
R_{1}+R_{2}  &  \leq\min\{I(X_{1}X_{2};Y_{1}|U)+C_{21},\text{ }I(X_{1}%
X_{2};Y_{2}\hat{Y}_{1}|U)\},\\
R_{0}+R_{1}+R_{2}  &  \leq\min\{I(X_{1}X_{2};Y_{1})+C_{21},\text{ }%
I(X_{1}X_{2};Y_{2}\hat{Y}_{1})\}\},
\end{align}
$\mathcal{R}_{TR,21\text{ }}$is similarly defined:
\end{subequations}
\begin{subequations}
\label{Rtr2}%
\begin{align}
\mathcal{R}_{TR,21}  &  =\{(R_{0},R_{1},R_{2})\text{:}\text{ }R_{0}%
\geq0,\text{ }R_{1}\geq0,\text{ }R_{2}\geq0,\\
R_{1}  &  \leq\min\{I(X_{1};Y_{1}\hat{Y}_{2}|X_{2}U),\text{ }I(X_{1}%
;Y_{2}|X_{2}U)+C_{12}\},\\
R_{2}  &  \leq\min\{I(X_{2};Y_{1}\hat{Y}_{2}|X_{1}U),\text{ }I(X_{1}%
;Y_{2}|X_{2}U)+C_{12}\},\\
R_{1}+R_{2}  &  \leq\min\{I(X_{1}X_{2};Y_{1}\hat{Y}_{2}|U),\text{ }%
I(X_{1};Y_{2}|X_{2}U)+C_{12}\},\\
R_{0}+R_{1}+R_{2}  &  \leq\min\{I(X_{1}X_{2};Y_{1}\hat{Y}_{2}),\text{ }%
I(X_{1};Y_{2}|X_{2})+C_{12}\}\},
\end{align}
subject to
\end{subequations}
\begin{align*}
C_{12}  &  \geq I(Y_{1};\hat{Y}_{1}|Y_{2})\\
C_{21}  &  \geq I(Y_{2};\hat{Y}_{2}|Y_{1})
\end{align*}
with $|\widehat{\mathcal{Y}}_{i}|\leq|\mathcal{Y}_{i}|+1$, and the union is
taken over all the joint distributions that factorize as
\[
p(u)p(x_{1}|u)p(x_{2}|u)p^{\ast}(y_{1}y_{2}|x_{1},x_{2})p(\hat{y}_{1}%
|y_{1})p(\hat{y}_{2}|y_{2}).
\]

\end{prop}

\begin{proof}
(Sketch): The proof is quite similar to Theorem 4 in \cite{Dabora Servetto
2006}, and here we only sketch the main points. Conferencing takes place via
$K=2$ conferencing rounds. Moreover, two possible strategies are considered,
giving rise to the convex hull operation in (\ref{Rts}) by time-sharing. The
two corresponding rate regions $\mathcal{R}_{TR,12}$ in (\ref{Rtr1}) and
$\mathcal{R}_{TR,21}$ in (\ref{Rtr2}) are obtained as follows. Consider
$\mathcal{R}_{TR,12}.$ Receiver 2 randomly partitions the message sets
$\mathcal{W}_{0},$ $\mathcal{W}_{1}$ and $\mathcal{W}_{2}$ into $2^{n\alpha
_{0}C_{12}}$, $2^{n\alpha_{1}C_{12}}$ and $2^{n\alpha_{2}C_{12}}$ subsets$, $
respectively, for a given $0\leq\alpha_{i}\leq1$ and $\sum_{i=0}^{2}\alpha
_{i}=1,$ as in the proof of Proposition \ref{prop2}.\ Encoding and
transmission are performed as for the MAC with common information. Receiver 1
compresses its received signal using Wyner-Ziv quantization as for the scheme
discussed in the proof of Proposition \ref{prop4}. This index is sent in the
first conferencing round (notice that $|\mathcal{V}_{1,1}|=nC_{12}$ and
$|\mathcal{V}_{2,1}|=0$). Upon reception of the compression index
$\mathcal{V}_{1,1}$, receiver 2 proceeds to decoding via joint typicality and
then sends the subset indices (see proof of Proposition \ref{prop2}) to
receiver 1 via $\mathcal{V}_{2,2}$ (now, $|\mathcal{V}_{1,2}|=0$ and
$|\mathcal{V}_{2,2}|=nC_{21}$). The latter decoder performs joint-typicality
decoding on the subsets of messages left undecided by the conferencing message
$\mathcal{V}_{1,1}$ received by 1. The rate region $\mathcal{R}_{TR,21}$ is
obtained similarly by simply swapping the roles of decoder 1 and decoder 2.
\end{proof}

\textit{Example 1 (cont'd)}: To see the impact of the two-round scheme, here
we reconsider Example 1 discussed above. It was shown that, for the scenario
discussed therein, the one-round scheme is not able to achieve the outer bound
$R_{out}$. However, it can be seen that the two-round scheme does indeed
achieve the outer bound. In fact, receiver 1 can provide $Y_{1}$ to receiver 2
via Slepian-Wolf compression as for the one-round case, while receiver 2 does
not send anything in the first conferencing round ($\hat{Y}_{2}$ is a
constant). Now, receiver 2 decodes and sends the bin index of the decoded
messages to receiver 1 in the second conferencing round according to the
two-round strategy discussed above (receiver 1 is silent in the second round).
Since $C_{21}\geq1/2\cdot I(X_{1}X_{2};Y_{2}|Y_{1})$ by assumption, it can be
seen from Proposition \ref{prop5} that the maximum equal rate achieved by the
two round scheme is $R_{TR}=R_{out}.$

We finally remark that it is possible in principle to extend the achievable
rate regions derived above to more than two conferencing rounds, following
\cite{Dabora Servetto 2006 bis} \cite{Draper et al 2003}. This is generally
advantageous in terms of achievable rates. While conceptually not difficult, a
description of the achievable rate region would require cumbersome notation
and is thus omitted here.

\section{Gaussian Compound MAC\label{sec_Gaussian}}

Here we consider the Gaussian version of the CM channel:
\begin{subequations}
\label{GCM}%
\begin{align}
Y_{1}  &  =\gamma_{11}X_{1}+\gamma_{21}X_{2}+Z_{1}\\
Y_{2}  &  =\gamma_{22}X_{2}+\gamma_{12}X_{1}+Z_{2},
\end{align}
with channel gains $\gamma_{ij}\geq0,$ independent white zero-mean unit-power
Gaussian noise $\{Z_{i}\}_{i=1}^{n}$ and per-symbol power constraints
$E[X_{i}^{2}]\leq P_{i}.$ Notice that the channel described by (\ref{GCM}) is
not physically degraded.

The outer bound of Proposition \ref{prop1} can be extended to (\ref{GCM}) by
using standard arguments. In particular, the capacity region of the Gaussian
CM, $\mathcal{C}_{CM}^{\mathcal{G}}(C_{12},C_{21})$ satisfies the following.
\end{subequations}
\begin{prop}
\label{prop6} We have $\mathcal{C}_{CM}^{\mathcal{G}}(C_{12},C_{21}%
)\subseteq\mathcal{C}_{CM-out}^{\mathcal{G}}(C_{12},C_{21})$ where:
\begin{subequations}
\label{outer bound Gaussian}%
\begin{align}
\mathcal{C}_{CM-out}^{\mathcal{G}}  &  (C_{12},C_{21}) =%
{\displaystyle\bigcup\limits_{\substack{0\leq P_{i}^{\prime}\leq P_{i}
\\i=1,2}}}
\bigg\{ (R_{0},R_{1},R_{2})\text{:} R_{0}\geq0,\text{ }R_{1}\geq0,\text{
}R_{2}\geq0,\\
R_{1}  &  \leq\min\left\{  \mathcal{C}(\gamma_{11}^{2}P_{1}^{\prime}%
)+C_{21},\text{ }\mathcal{C}(\gamma_{12}^{2}P_{1}^{\prime})+C_{12},\text{
}\mathcal{C}(P_{1}^{\prime}\left(  \gamma_{11}^{2}+\gamma_{12}^{2}\right)
)\right\}  ,\\
R_{2}  &  \leq\min\left\{  \mathcal{C}(\gamma_{21}^{2}P_{2}^{\prime}%
)+C_{21},\text{ }\mathcal{C}(\gamma_{22}^{2}P_{2}^{\prime})+C_{12},\text{
}\mathcal{C}(P_{2}^{\prime}\left(  \gamma_{21}^{2}+\gamma_{22}^{2}\right)
)\right\}  ,\\
R_{1}+R_{2}  &  \leq\min\left\{
\begin{array}
[c]{l}%
\mathcal{C}(\gamma_{11}^{2}P_{1}^{\prime}+\gamma_{21}^{2}P_{2}^{\prime
})+C_{21},\text{ }\mathcal{C}(\gamma_{22}^{2}P_{2}^{\prime}+\gamma_{12}%
^{2}P_{1}^{\prime})+C_{12},\\
\mathcal{C}\left(  P_{1}^{\prime}\left(  \gamma_{11}^{2}+\gamma_{12}%
^{2}\right)  +P_{2}^{\prime}\left(  \gamma_{21}^{2}+\gamma_{22}^{2}\right)
+\mathcal{K}\right)
\end{array}
\right\} \\
R_{0}+R_{1}+R_{2}  &  \leq\min\left\{
\begin{array}
[c]{c}%
\mathcal{C}(\gamma_{11}^{2}P_{1}^{\prime}+\gamma_{21}^{2}P_{2}^{\prime}%
+\rho_{1})+C_{21},\text{ }\mathcal{C}(\gamma_{22}^{2}P_{2}^{\prime}%
+\gamma_{12}^{2}P_{1}^{\prime}+\rho_{2})+C_{12},\\
\mathcal{C}\left(
\begin{array}
[c]{l}%
P_{1}^{\prime}\left(  \gamma_{11}^{2}+\gamma_{12}^{2}\right)  +P_{2}^{\prime
}\left(  \gamma_{21}^{2}+\gamma_{22}^{2}\right)  +\mathcal{K}\\
+\rho_{1}\left(  1+P_{1}^{\prime}\gamma_{12}^{2}+P_{2}^{\prime}\gamma_{22}%
^{2}\right)  +\rho_{2}(1+P_{1}^{\prime}\gamma_{11}^{2}+P_{2}^{\prime}%
\gamma_{21}^{2})\\
-2\sqrt{\rho_{1}\rho_{2}}(P_{1}^{\prime}\gamma_{11}\gamma_{12}+P_{2}^{\prime
}\gamma_{21}\gamma_{22})
\end{array}
\right)
\end{array}
\right\}  \bigg\} .\nonumber
\end{align}
with
\end{subequations}
\begin{subequations}
\label{rho function}%
\begin{align}
\mathcal{K}  &  =P_{1}^{\prime}P_{2}^{\prime}(\gamma_{12}\gamma_{21}%
-\gamma_{11}\gamma_{22})^{2}\\
\rho_{1}  &  =(\gamma_{11}\sqrt{P_{1}-P_{1}^{\prime}}+\gamma_{21}\sqrt
{P_{2}-P_{2}^{\prime}})^{2}\\
\rho_{2}  &  =(\gamma_{22}\sqrt{P_{2}-P_{2}^{\prime}}+\gamma_{12}\sqrt
{P_{1}-P_{1}^{\prime}})^{2}%
\end{align}
and $\mathcal{C}(x)\triangleq\frac{1}{2}\log(1+x).$
\end{subequations}
\end{prop}

\begin{proof}
Similarly to Proposition \ref{prop1}, one can prove that the rate region
(\ref{outer bound}) is an outer bound on the achievable rates. It then remains
to be proved that a Gaussian joint distribution $p(u)p(x_{1}|u)p(x_{2}|u)$
with $X_{i}=\sqrt{P-P_{i}^{\prime}}U+\sqrt{P_{i}^{\prime}}V_{i},$ where is
$U,$ $V_{1}$ and $V_{2}$ are independent Gaussian zero-mean unit-power random
variables, is optimal. This can be done following the steps of \cite{Lapidoth}%
, where the proof is given for a single MAC\ channel with common information
(see also \cite{Simeone Asilomar}). The proof is concluded with some algebra.
\end{proof}

The achievable rates in Proposition \ref{prop4} (for $K=1$) and Proposition
\ref{prop5} (for $K=2$) can also be extended to the Gaussian CM channel. In so
doing, we focus on jointly Gaussian auxiliary random variables for Wyner-Ziv
compression. While no general claim of optimality is put forth here, some
conclusion on the optimality of such schemes can be drawn as discussed later
in Sec. \ref{sec_discussion}.

\begin{prop}
\label{prop7} The following rate region is achievable with one-round
conferencing, $K=1$:
\begin{subequations}
\label{one step Gaussian}%
\begin{align}
\mathcal{R}_{OR}^{\mathcal{G}}(C_{12},C_{21}) &  =%
{\displaystyle\bigcup\limits_{\substack{0\leq P_{i}^{\prime}\leq P_{i}%
\\i=1,2}}}
\{(R_{0},R_{1},R_{2})\text{:}\text{ }R_{0}\geq0,\text{ }R_{1}\geq0,\text{
}R_{2}\geq0,\\
R_{1} &  \leq\min\left\{  \mathcal{C}\left(  P_{1}^{\prime}\left(  \gamma
_{11}^{2}+\frac{\gamma_{12}^{2}}{1+\sigma_{2}^{2}}\right)  \right)  ,\text{
}\mathcal{C}\left(  P_{1}^{\prime}\left(  \gamma_{12}^{2}+\frac{\gamma
_{11}^{2}}{1+\sigma_{1}^{2}}\right)  \right)  \right\}  ,\\
R_{2} &  \leq\min\left\{  \mathcal{C}\left(  P_{2}^{\prime}\left(  \gamma
_{21}^{2}+\frac{\gamma_{22}^{2}}{1+\sigma_{2}^{2}}\right)  \right)  ,\text{
}\mathcal{C}\left(  P_{2}^{\prime}\left(  \gamma_{22}^{2}+\frac{\gamma
_{21}^{2}}{1+\sigma_{1}^{2}}\right)  \right)  \right\}  \\
R_{1}+R_{2} &  \leq\min\left\{
\begin{array}
[c]{c}%
\mathcal{C}\left(
\begin{array}
[c]{l}%
P_{1}^{\prime}\left(  \gamma_{11}^{2}+\frac{\gamma_{12}^{2}}{1+\sigma_{2}^{2}%
}\right)  +P_{2}^{\prime}\left(  \gamma_{21}^{2}+\frac{\gamma_{22}^{2}%
}{1+\sigma_{2}^{2}}\right)  \\
+\frac{\mathcal{K}}{1+\sigma_{2}^{2}}%
\end{array}
\right)  ,\text{ }\\
\mathcal{C}\left(
\begin{array}
[c]{l}%
P_{2}^{\prime}\left(  \gamma_{22}^{2}+\frac{\gamma_{21}^{2}}{1+\sigma_{1}^{2}%
}\right)  +P_{1}^{\prime}\left(  \gamma_{21}^{2}+\frac{\gamma_{11}^{2}%
}{1+\sigma_{1}^{2}}\right)  \\
+\frac{\mathcal{K}}{1+\sigma_{1}^{2}}%
\end{array}
\right)
\end{array}
\right\}  \\
R_{0}+R_{1}+R_{2} &  \leq\min\left\{
\begin{array}
[c]{c}%
\mathcal{C}\left(
\begin{array}
[c]{l}%
P_{1}^{\prime}\left(  \gamma_{11}^{2}+\frac{\gamma_{12}^{2}}{1+\sigma_{2}^{2}%
}\right)  +P_{2}^{\prime}\left(  \gamma_{21}^{2}+\frac{\gamma_{22}^{2}%
}{1+\sigma_{2}^{2}}\right)  +\frac{\mathcal{K}}{1+\sigma_{2}^{2}}\\
+\rho_{1}\left(  1+\frac{P_{1}^{\prime}\gamma_{12}^{2}+P_{2}^{\prime}%
\gamma_{22}^{2}}{1+\sigma_{2}^{2}}\right)  +\frac{\rho_{2}}{1+\sigma_{2}^{2}%
}(1+P_{1}^{\prime}\gamma_{11}^{2}+P_{2}^{\prime}\gamma_{21}^{2})\\
-\frac{2\sqrt{\rho_{1}\rho_{2}}(P_{1}\gamma_{11}\gamma_{12}+P_{2}\gamma
_{21}\gamma_{22})}{1+\sigma_{2}^{2}}%
\end{array}
\right)  ,\\
\mathcal{C}\left(
\begin{array}
[c]{l}%
P_{1}^{\prime}\left(  \gamma_{12}^{2}+\frac{\gamma_{11}^{2}}{1+\sigma_{1}^{2}%
}\right)  +P_{2}^{\prime}\left(  \gamma_{22}^{2}+\frac{\gamma_{21}^{2}%
}{1+\sigma_{1}^{2}}\right)  +\frac{\mathcal{K}}{1+\sigma_{1}^{2}}\\
+\rho_{2}\left(  1+\frac{P_{1}^{\prime}\gamma_{11}^{2}+P_{2}^{\prime}%
\gamma_{21}^{2}}{1+\sigma_{1}^{2}}\right)  +\frac{\rho_{1}}{1+\sigma_{1}^{2}%
}(1+P_{1}^{\prime}\gamma_{12}^{2}+P_{2}^{\prime}\gamma_{22}^{2})\\
-\frac{2\sqrt{\rho_{1}\rho_{2}}(P_{1}\gamma_{11}\gamma_{12}+P_{2}\gamma
_{21}\gamma_{22})}{1+\sigma_{1}^{2}}%
\end{array}
\right)
\end{array}
\right\}  \},\nonumber
\end{align}
with (\ref{rho function}) and quantization noise variances satisfying
\end{subequations}
\begin{subequations}
\label{sigmas1}%
\begin{align}
\sigma_{1}^{2}  & \geq\frac{1+(\gamma_{11}^{2}+\gamma_{12}^{2})P_{1}%
+(\gamma_{21}^{2}+\gamma_{22}^{2})P_{2}+(\gamma_{12}\gamma_{21}-\gamma
_{11}\gamma_{22})^{2}P_{1}P_{2}}{(2^{2C_{12}}-1)(1+\gamma_{12}^{2}P_{1}%
+\gamma_{22}^{2}P_{2})}\\
\sigma_{2}^{2}  & \geq\frac{1+(\gamma_{11}^{2}+\gamma_{12}^{2})P_{1}%
+(\gamma_{21}^{2}+\gamma_{22}^{2})P_{2}+(\gamma_{12}\gamma_{21}-\gamma
_{11}\gamma_{22})^{2}P_{1}P_{2}}{(2^{2C_{21}}-1)(1+\gamma_{11}^{2}P_{1}%
+\gamma_{21}^{2}P_{2})}.
\end{align}

\end{subequations}
\end{prop}

\begin{proof}
As stated above, we consider Gaussian auxiliary random variables and evaluate
the region (\ref{OS}). In particular, the test channels for Wyner-Ziv
compression are selected as $\hat{Y}_{i}=Y_{i}+Z_{q,i}$ where the compression
noise $Z_{q,i}$\ is zero-mean Gaussian with variance $\sigma_{i}^{2}$ and
independent of $Y_{i}$. The proposition follows from some algebraic manipulation.
\end{proof}

The one-round strategy can be generalized by enabling two rounds of
conferencing ($K=2$), obtaining the following achievable rate region:

\begin{prop}
\label{prop8} The following rate region is achievable with two rounds of
conferencing, $K=2$:
\begin{equation}
\mathcal{R}_{TR}^{\mathcal{G}}(C_{12},C_{21})=\text{co}%
{\displaystyle\bigcup\limits_{\substack{0\leq P_{i}^{\prime}\leq P_{i}%
\\i=1,2}}}
\{\mathcal{R}_{TR,12}^{\mathcal{G}}\cup\mathcal{R}_{TR,21}^{\mathcal{G}%
}\}\label{two step Gaussian}%
\end{equation}
with
\begin{subequations}
\begin{align}
\mathcal{R}_{TR,21}^{\mathcal{G}} &  =\bigg\{(R_{0},R_{1},R_{2})\text{:}\text{
}R_{0}\geq0,\text{ }R_{1}\geq0,\text{ }R_{2}\geq0,\\
R_{1} &  \leq\min\left\{  \mathcal{C}\left(  \gamma_{11}^{2}P_{1}^{\prime
}\right)  +C_{21},\text{ }\mathcal{C}\left(  P_{1}^{\prime}\left(  \gamma
_{12}^{2}+\frac{\gamma_{11}^{2}}{1+\sigma_{1}^{2}}\right)  \right)  \right\}
\\
R_{2} &  \leq\min\left\{  \mathcal{C}\left(  \gamma_{21}^{2}P_{2}^{\prime
}\right)  +C_{21},\text{ }\mathcal{C}\left(  P_{2}^{\prime}\left(  \gamma
_{22}^{2}+\frac{\gamma_{21}^{2}}{1+\sigma_{1}^{2}}\right)  \right)  \right\}
\\
R_{1}+R_{2} &  \leq\min\left\{
\begin{array}
[c]{l}%
\mathcal{C}\left(  \gamma_{11}^{2}P_{1}^{\prime}+\gamma_{21}^{2}P_{2}^{\prime
}\right)  +C_{21},\text{ }\\
\mathcal{C}\left(
\begin{array}
[c]{l}%
P_{1}^{\prime}\left(  \gamma_{12}^{2}+\frac{\gamma_{11}^{2}}{1+\sigma_{1}^{2}%
}\right)  +P_{2}^{\prime}\left(  \gamma_{22}^{2}+\frac{\gamma_{21}^{2}%
}{1+\sigma_{1}^{2}}\right)  +\frac{\mathcal{K}}{1+\sigma_{1}^{2}}%
\end{array}
\right)
\end{array}
\right\}  \\
R_{0}+R_{1}+R_{2} &  \leq\min\left\{
\begin{array}
[c]{l}%
\mathcal{C}(\gamma_{11}^{2}P_{1}^{\prime}+\gamma_{21}^{2}P_{2}^{\prime}%
+\rho_{1})+C_{21},\\
\mathcal{C}\left(
\begin{array}
[c]{l}%
P_{1}^{\prime}\left(  \gamma_{12}^{2}+\frac{\gamma_{11}^{2}}{1+\sigma_{1}^{2}%
}\right)  +P_{2}^{\prime}\left(  \gamma_{22}^{2}+\frac{\gamma_{21}^{2}%
}{1+\sigma_{1}^{2}}\right)  +\frac{\mathcal{K}}{1+\sigma_{1}^{2}}\\
+\rho_{2}\left(  1+\frac{P_{1}^{\prime}\gamma_{11}^{2}+P_{2}^{\prime}%
\gamma_{21}^{2}}{1+\sigma_{1}^{2}}\right)  +\frac{\rho_{1}}{1+\sigma_{1}^{2}%
}(1+P_{1}^{\prime}\gamma_{12}^{2}+P_{2}^{\prime}\gamma_{22}^{2})\\
-\frac{2\sqrt{\rho_{1}\rho_{2}}(P_{1}^{\prime}\gamma_{11}\gamma_{12}%
+P_{2}^{\prime}\gamma_{21}\gamma_{22})}{1+\sigma_{1}^{2}}%
\end{array}
\right)
\end{array}
\right\}  \bigg\}\nonumber
\end{align}
$\mathcal{R}_{TR,12\text{ }}$is similarly defined:
\end{subequations}
\begin{subequations}
\begin{align}
\mathcal{R}_{TR,12}^{\mathcal{G}} &  =\bigg\{(R_{0},R_{1},R_{2})\text{:}\text{
}R_{0}\geq0,\text{ }R_{1}\geq0,\text{ }R_{2}\geq0,\\
R_{1} &  \leq\min\left\{  \mathcal{C}\left(  P_{1}^{\prime}\left(  \gamma
_{11}^{2}+\frac{\gamma_{12}^{2}}{1+\sigma_{2}^{2}}\right)  \right)  ,\text{
}\mathcal{C}\left(  \gamma_{12}^{2}P_{1}^{\prime}\right)  +C_{12}\right\}  \\
R_{2} &  \leq\min\left\{  \mathcal{C}\left(  P_{2}^{\prime}\left(  \gamma
_{21}^{2}+\frac{\gamma_{22}^{2}}{1+\sigma_{2}^{2}}\right)  \right)  ,\text{
}\mathcal{C}\left(  \gamma_{22}^{2}P_{2}^{\prime}\right)  +C_{12}\right\}  \\
R_{1}+R_{2} &  \leq\min\left\{
\begin{array}
[c]{l}%
\mathcal{C}\left(
\begin{array}
[c]{c}%
P_{1}^{\prime}\left(  \gamma_{11}^{2}+\frac{\gamma_{12}^{2}}{1+\sigma_{2}^{2}%
}\right)  +P_{2}^{\prime}\left(  \gamma_{21}^{2}+\frac{\gamma_{22}^{2}%
}{1+\sigma_{2}^{2}}\right)  \\
+\frac{P_{1}^{\prime}P_{2}^{\prime}(\gamma_{12}\gamma_{21}-\gamma_{11}%
\gamma_{22})^{2}}{1+\sigma_{2}^{2}}%
\end{array}
\right)  ,\\
\mathcal{C}\left(  \gamma_{22}^{2}P_{2}^{\prime}+\gamma_{12}^{2}P_{1}^{\prime
}\right)  +C_{12}%
\end{array}
\right\}  
\end{align}
\begin{align}
R_{0}+R_{1}+R_{2} &  \leq\min\left\{
\begin{array}
[c]{l}%
\mathcal{C}\left(
\begin{array}
[c]{l}%
P_{1}^{\prime}\left(  \gamma_{11}^{2}+\frac{\gamma_{12}^{2}}{1+\sigma_{2}^{2}%
}\right)  +P_{2}^{\prime}\left(  \gamma_{21}^{2}+\frac{\gamma_{22}^{2}%
}{1+\sigma_{2}^{2}}\right)  +\frac{\mathcal{K}}{1+\sigma_{2}^{2}}\\
+\rho_{1}\left(  1+\frac{P_{1}^{\prime}\gamma_{12}^{2}+P_{2}^{\prime}%
\gamma_{22}^{2}}{1+\sigma_{2}^{2}}\right)  +\frac{\rho_{2}}{1+\sigma_{2}^{2}%
}(1+P_{1}^{\prime}\gamma_{11}^{2}+P_{2}^{\prime}\gamma_{21}^{2})\\
-\frac{2\sqrt{\rho_{1}\rho_{2}}(P_{1}\gamma_{11}\gamma_{12}+P_{2}\gamma
_{21}\gamma_{22})}{1+\sigma_{2}^{2}}%
\end{array}
\right)  ,\\
\mathcal{C}(\gamma_{22}^{2}P_{2}^{\prime}+\gamma_{12}^{2}P_{1}^{\prime}%
+\rho_{2})+C_{12}%
\end{array}
\right\}  \bigg\},\nonumber
\end{align}
with (\ref{rho function}) and (\ref{sigmas1}).
\end{subequations}
\end{prop}

\subsection{Discussion}

\label{sec_discussion}

Here we draw some conclusions on the optimality of the one and two-round
schemes discussed above for the Gaussian CM channel. We start with the
one-round scheme of Proposition \ref{prop7} and notice that, by comparison
with the outer bound (\ref{outer bound Gaussian}), it can be easily seen that
the scheme at hand is optimal in the asymptotic regime of large conferencing
capacities $C_{12}\rightarrow\infty$ and $C_{21}\rightarrow\infty.$ Further
conclusions on the gap between the upper bound (\ref{outer bound Gaussian})
and the performance achievable with one round of conferencing at the decoders
can be drawn in two special cases. Consider first the case of a broadcast
channel with conferencing encoders \cite{Dabora Servetto 2006} \cite{Draper et
al 2003}, which is obtained as $R_{0}=0$ and $R_{2}=0$ and thus $P_{2}=0$
without loss of generality (a symmetric statement can be straightforwardly
obtained for $R_{0}=0$ and $R_{1}=0$). In this case, we show below that the
one-round scheme achieves the upper bound (\ref{outer bound Gaussian}) to
within half a bit, irrespective of the channel gains of the broadcast channel
and the capacities of the conferencing links. To elaborate, notice that the
outer bound (\ref{outer bound Gaussian}) for the case at hand is given by
\begin{equation}
R_{1}\leq R_{1,out}=\min\{\mathcal{C}(\gamma_{11}^{2}P_{1})+C_{21}%
,\mathcal{C}(\gamma_{12}^{2}P_{1})+C_{12},\mathcal{C}((\gamma_{11}^{2}%
+\gamma_{12}^{2})P_{1})\}, \label{R1out}%
\end{equation}
whereas the rate achievable with one-round conferencing is given by
\begin{equation}
R_{1,OR}=\min\left\{  \mathcal{C}\left(  \gamma_{11}^{2}P_{1}+\frac
{\gamma_{12}^{2}P_{1}}{1+\sigma_{2}^{2}}\right)  ,\mathcal{C}\left(
\gamma_{12}^{2}P_{1}+\frac{\gamma_{11}^{2}P_{1}}{1+\sigma_{1}^{2}}\right)
\right\}  , \label{R1or}%
\end{equation}
where
\begin{align}
\sigma_{1}^{2}  &  = \frac{1+(\gamma_{11}^{2} + \gamma_{12}^{2}) P_{1}%
}{(2^{2C_{12}}-1)(1+\gamma_{12}^{2}P_{1})} ,\nonumber
\end{align}
and
\begin{align}
\sigma_{2}^{2}  &  = \frac{1+(\gamma_{11}^{2} + \gamma_{12}^{2}) P_{1}%
}{(2^{2C_{21}}-1)(1+\gamma_{11}^{2}P_{1})}.\nonumber
\end{align}
Using these two expressions, we can prove the following proposition (see
Appendix \ref{app:9} for a full proof).

\begin{prop}
\label{prop9} We have $R_{1,OR}\geq R_{1,out}-\frac{1}{2}.$ Moreover, for the
symmetric channel case, i.e., $\gamma_{11}^{2}=\gamma_{12}^{2}$, we have
$R_{1,OR}\geq R_{1,out}-\frac{\log3-1}{2}.$
\end{prop}

Next, we consider the symmetric Gaussian CM channel, that is, we let $R_{0}%
=0$, $\gamma_{11}^{2}=\gamma_{22}^{2}=a$, $\gamma_{12}^{2}=\gamma_{21}^{2}=b$,
and $P_{1}=P_{2}\triangleq P$. We also assume symmetric conferencing link
capacities $C_{12}=C_{21}\triangleq C$. In such a case, the outer bound and
the achievable rates with one-round conferencing are:
\begin{subequations}
\begin{align}
\mathcal{C}_{CM-out}^{\mathcal{G}}(C)  &  =\{(R_{1},R_{2}):R_{1}\geq0,\text{
}R_{2}\geq0,\nonumber\\
R_{1}  &  \leq\min\{\mathcal{C}(aP)+C,~\mathcal{C}(bP)+C,~\mathcal{C}%
((a+b)P)\},\\
R_{2}  &  \leq\min\{\mathcal{C}(bP)+C,~\mathcal{C}(aP)+C,~\mathcal{C}%
((a+b)P)\},\\
R_{1}+R_{2}  &  \leq\min\{\mathcal{C}((a+b)P)+C,\text{ }\mathcal{C}%
(2(a+b)P+(b-a)^{2}P^{2})\} \},
\end{align}
and
\end{subequations}
\begin{subequations}
\begin{align}
\mathcal{R}_{OR}^{\mathcal{G}}(C)  &  =\left\{  (R_{1},R_{2}):R_{1}%
\geq0,\text{ }R_{2}\geq0,\right. \nonumber\\
R_{1}  &  \leq\min\left\{  \mathcal{C}\left(  \left(  a+\frac{b}{1+\sigma^{2}%
}\right)  P\right)  ,\mathcal{C}\left(  \left(  b+\frac{a}{1+\sigma^{2}%
}\right)  P\right)  \right\}  ,\\
R_{2}  &  \leq\min\left\{  \mathcal{C}\left(  \left(  b+\frac{a}{1+\sigma^{2}%
}\right)  P\right)  ,\mathcal{C}\left(  \left(  a+\frac{b}{1+\sigma^{2}%
}\right)  P\right)  \right\}  ,\\
&  \left.  R_{1}+R_{2}\leq\mathcal{C}\left(  \left(  a+b+\frac{a+b}%
{1+\sigma^{2}}+\frac{(b-a)^{2}P}{1+\sigma^{2}}\right)  P\right)  \right\}  ,
\end{align}
with $\sigma^{2}\triangleq\frac{1+2(a+b)P+(b-a)^{2}P^{2}}{(1+(a+b)P)(2^{2C}%
-1)},$ respectively. The following result can be proved (see Appendix
\ref{app:10}).
\end{subequations}
\begin{prop}
\label{prop10}$\mathcal{R}_{OR}^{\mathcal{G}}\supseteq\{(R_{1},R_{2}%
):R_{1}\geq0,R_{2}\geq0,(R_{1}+\delta,R_{2}+(\Delta-\delta))\in\mathcal{C}%
_{CM-out}^{\mathcal{G}}(C)\mbox{ for all }\delta\in\lbrack0,\Delta]\}$ with
$\Delta=\frac{\log(1+\beta)}{2}$ where $\beta\triangleq\frac{\max(a,b)}%
{\min(a,b)}$. Moreover, in the special case $a=b$, the gap $\Delta$ can be
further reduced to $\Delta=\left(  \frac{\log3-1}{2}\right)  \approx0.293$ bits.
\end{prop}

The proposition above is equivalent to saying that the total rate loss of
using one round of conferencing relative to the sum capacity is less than
$\frac{\log(1+\beta)}{2}$, which is a constant that depends only on the
relative qualities of the direct channels and the cross channels.

Let us now consider the two-round scheme of Proposition \ref{prop8}. Since
$\mathcal{R}_{TR}^{\mathcal{G}}(C_{12},C_{21})\supseteq\mathcal{R}%
_{OR}^{\mathcal{G}}(C_{12},C_{21}),$ all the conclusions above on the
one-round scheme apply also to the two-round strategy. Alternatively, we can
interpret these results as a finite bit limit on the potential gain of going
from one round of conferencing to two rounds. Moreover, it should be noted
that the two-round approach was defined as single-session in \cite{Steiner}
and shown therein to be optimal among several classes of multi-session
protocols for a broadcast channel with cooperating decoders. Finally, we can
prove the following.

\begin{prop}
The two-round scheme is optimal in the case of \textit{unidirectional
cooperation}: $\mathcal{R}_{TR}^{\mathcal{G}}(0,C_{21})=\mathcal{C}%
_{CM-out}^{\mathcal{G}}(0,C_{21})$ and $\mathcal{R}_{TR}^{\mathcal{G}}%
(C_{12},0)=\mathcal{C}_{CM-out}^{\mathcal{G}}(C_{12},0),$ thus establishing
the capacity of the Gaussian CM channel for this special case.
\end{prop}

\begin{proof}
This result follows by comparing the achievable region with the outer bound
(\ref{outer bound Gaussian}).
\end{proof}

Next, we comment on the \textit{sum-rate multiplexing gain} of the Gaussian CM
channel. Consider a symmetric system with $P_{1}=P_{2}\triangleq P,$
$\gamma_{11}=\gamma_{22},$ $\gamma_{12}=\gamma_{21},$ and $C_{12}%
=C_{21}\triangleq C.$ We are interested in studying the conditions on the
conferencing capacity $C$ such that the maximum multiplexing gain on the
sum-rate, $\lim_{P\rightarrow\infty}\sup_{(0,R_{1},R_{2})\in\mathcal{C}%
_{CM}^{\mathcal{G}}(C,C)}(R_{1}+R_{2})/(\frac{1}{2}\log P)=2,$\ corresponding
to full cooperation, can be achieved. From the outer bound in
(\ref{outer bound Gaussian}), it is clear that $C$ should scale at least as
$\frac{1}{2}\log P$ as the sum rate is limited by $\mathcal{C}(P(\gamma
_{11}^{2}+\gamma_{21}^{2}))+C.$ By considering the achievable regions with one
(\ref{one step Gaussian}) or two (\ref{two step Gaussian}) conferencing
rounds, it can be also concluded that if $C$ scales as $(1+\epsilon)\log P$
with any $\epsilon>0$, then the optimal multiplexing gain is indeed
achievable. This is because with $C=\frac{1}{2}(1+\epsilon)\log P$ the
quantization noise variances in (\ref{sigmas1}) are proportional to
$P^{-\epsilon}$ and thus tend to zero for large $P.$ It is noted that this
result would hold even if the decoders used regular compression that neglects
the side information at the other decoder, as in this case we would have
$\sigma_{i}^{2}=\frac{\gamma_{11}^{2}P+\gamma_{21}^{2}P+1}{2^{2C}-1},$ which
is still proportional to $P^{-\epsilon}$ for $C=\frac{1}{2}(1+\epsilon)\log
P.$

As a final remark, extending the achievable rates defined above for the
Gaussian channel (and assuming Gaussian channel and compression codebooks as
done above) to more than two conferencing rounds would not lead to any further
gain, since with Gaussian variables, ``conditional'' compression and
compression with side information have the same efficiency (see \cite{Draper
et al 2003} for a discussion).

\subsection{Numerical results}

Since the rate region expressions provided for the outer bound and the
one-round and two-round achievable schemes give little insight, in this
section we present numerical results to see how much gain is obtained via
decoder cooperation. In Fig. \ref{fig1}, we consider a symmetric scenario with
$P_{1}=P_{2}=5~\mathrm{dB}$, $\gamma_{12}^{2}=\gamma_{21}^{2}=-3~\mathrm{dB}$,
$\gamma_{11}^{2}=\gamma_{22}^{2}=0~\mathrm{dB}$, $C_{21}=C_{12}=0.3$, and we
plot the outer bound (\ref{outer bound Gaussian}), the rate region achievable
with one-round (\ref{one step Gaussian}) and two-round
(\ref{two step Gaussian}) conferencing as well as with no cooperation
($C_{12}=C_{21}=0$) (obtained from either (\ref{one step Gaussian}) and
(\ref{two step Gaussian})) for $R_{0}=0$ (so that selecting $P_{i}^{\prime
}=P_{i}$ is sufficient in all the capacity regions). It can be seen that
cooperation via conferencing decoders enables the achievable rate region to be
increased both in terms of sum-rate and individual rates. Moreover, the
two-step strategy provides relevant gains with respect to the one-step
approach, while still not achieving the outer bound
(\ref{outer bound Gaussian}).%

\begin{figure}
\begin{center}
\includegraphics[width=4.0421in]%
{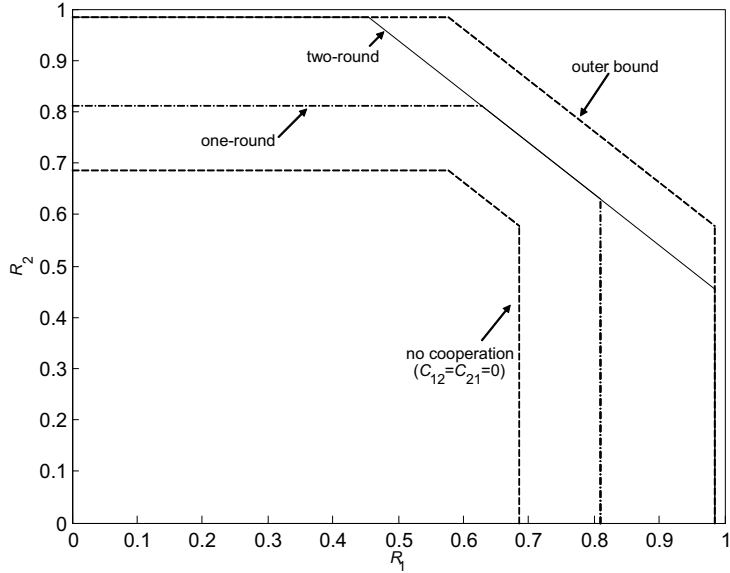}%
\caption{Outer bound (\ref{outer bound Gaussian}), rate region achievable with
one-round ( \ref{one step Gaussian}) and two-round (\ref{two step Gaussian})
strategies and with no cooperation ($C_{12}=C_{21}=0$) for $R_{0}=0,$ and a
symmetric scenario with $P_{1}=P_{2}=5dB,$ $\gamma_{12}^{2}=\gamma_{21}%
^{2}=-3dB,$ $\gamma_{11}^{2}=\gamma_{22}^{2}=0dB,$ $C_{21}=C_{12}=0.5$.}%
\label{fig1}%
\end{center}
\end{figure}
\begin{figure}
\begin{center}
\includegraphics[width=4.0421in]%
{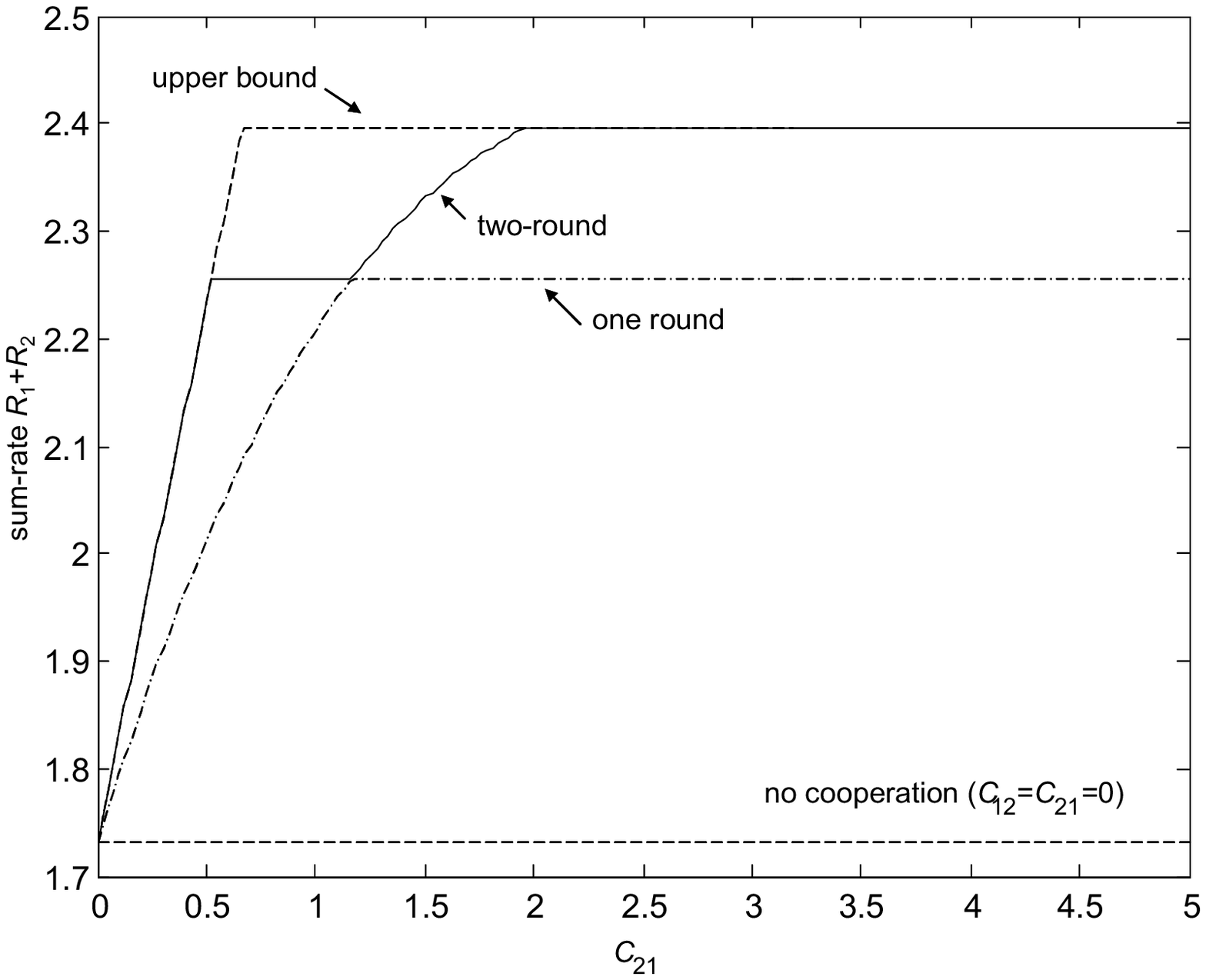}%
\caption{Sum of the private rates $R_{1}+R_{2}$ (with $R_{0}=0)$ versus the
conferencing link capacity $C_{21}$ for the outer bound
(\ref{outer bound Gaussian}), the one-round (\ref{one step Gaussian}) and
two-round (\ref{two step Gaussian}) strategies and with no cooperation
($P_{1}=P_{2}=10dB,$ $\gamma_{12}^{2}=0dB,$ $\gamma_{22}^{2}=0dB,$
$\gamma_{21}^{2}=-3dB,$ $\gamma_{11}^{2}=-3dB,$ $C_{12}=0.2).$}%
\label{fig2}%
\end{center}
\end{figure}
\begin{figure}
\begin{center}
\includegraphics[width=4.0421in]%
{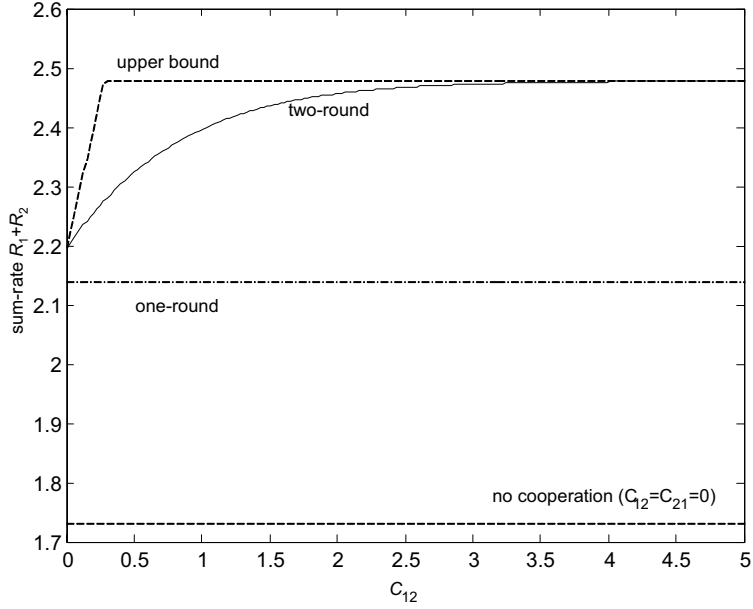}%
\caption{Sum of the private rates $R_{1}+R_{2}$ (with $R_{0}=0)$ versus the
conferencing link capacity $C_{12}$ for the outer bound
(\ref{outer bound Gaussian}), the one-round (\ref{one step Gaussian}) and
two-round (\ref{two step Gaussian}) strategies and with no cooperation
($P_{1}=P_{2}=10dB,$ $\gamma_{12}^{2}=0dB,$ $\gamma_{22}^{2}=0dB,$
$\gamma_{21}^{2}=-3dB,$ $\gamma_{11}^{2}=-3dB,$ $C_{21}=0.8).$}%
\label{fig3}%
\end{center}
\end{figure}

Fig. \ref{fig2} and Fig. \ref{fig3} show the sum of the private rates
$R_{1}+R_{2} $ (with $R_{0}=0$) versus the conferencing link capacities
$C_{21}$ and $C_{12}$, respectively, for the outer bound
(\ref{outer bound Gaussian}), the achievable schemes with one-round
(\ref{one step Gaussian}) and two-round (\ref{two step Gaussian}) conferencing
and with no cooperation. In both figures, we consider cases in which receiver
1 has a worse signal quality than receiver 2 (stochastically degraded):
$P_{1}=P_{2}=10dB,$ $\gamma_{12}^{2}=0dB,$ $\gamma_{22}^{2}=0dB,$ $\gamma
_{21}^{2}=-3dB,$ $\gamma_{11}^{2}=-3dB.$ Fig. \ref{fig2} shows the achievable
sum-rates versus $C_{21}$ for $C_{12}=0.2.$ It is seen that if $C_{21}=0$ the
upper bound coincides with the rate achievable with no cooperation, showing
that if the link from the "good" receiver to the degraded receiver is
disabled, the performance is dominated by the worse receiver and there is no
gain in having $C_{12}>0.$ Increasing $C_{21}$ enables the rate of the worse
receiver to be increased via cooperation, thus harnessing significant gains
with respect to no cooperation. In particular, it is seen that for $C_{21}$
sufficiently small (here $C_{21}\lesssim0.5$) the two-step strategy is
optimal, since in this region the performance is dominated by the worse
receiver whose achievable rate increases linearly with $C_{21}$ due to
cooperation via binning of the message set performed at the good receiver. The
one-step protocol instead lags behind and its performance saturates at
$\mathcal{C}\left(  \gamma_{22}^{2}P_{2}+\gamma_{21}^{2}P_{1}+\frac
{\gamma_{11}^{2}P_{1}+\gamma_{21}^{2}P_{2}}{1+\sigma_{1}^{2}}\right)
\simeq2.26.$ Finally, for sufficiently large $C_{21},$ the achievable sum-rate
at the worse receiver becomes larger than $2.26$ and the performance tends to
the sum-rate of the best receiver, $\mathcal{C}\left(  \gamma_{22}^{2}%
P_{2}+\gamma_{12}^{2}P_{1}\right)  +C_{12}\simeq2.4,$ unless $C_{12}$ is too large.

Further insight is shown in Fig. \ref{fig3} where the rates are plotted versus
$C_{12}$ for $C_{21}=0.8$. We notice that for $C_{12}=0$ only the two-step
protocol is able to achieve the upper bound, since in this regime it is
optimal for the good receiver to decode and bin its decision. Moreover,
similarly, increasing $C_{12}$ enhances the gain of the two-round strategy
over the one-round strategy up to the point where the perfomance is limited by
the sum-rate at the worse receiver, i.e., by $\mathcal{C}\left(  P_{1}\left(
\gamma_{11}^{2}+\gamma_{12}^{2}\right)  +P_{2}\left(  \gamma_{21}^{2}%
+\gamma_{22}^{2}\right)  +P_{1}P_{2}(\gamma_{12}\gamma_{21}-\gamma_{11}%
\gamma_{22})^{2}\right)  \simeq2.48,$ which coincides with the upper bound.

\section{Conferencing encoders and decoders\label{sec_enc dec}}

In this section, we extend the capacity results of Sec. \ref{sec_capacity_1}
to the scenario in Fig. \ref{model1} in which instead of having a common
message (as in the previous sections), the encoders are connected via
conferencing links of capacity $\bar{C}_{12}$ and $\bar{C}_{21}.$ Here, each
encoder has only one message $W_{i}$ of rate $R_{i}$\ ($i=1,2$) to deliver to
both decoders. We refer to this channel as a compound MAC with conferencing
decoders and encoders (for short, the CME channel). Definitions of encoders
and conferencing at the transmission side follows the standard reference
\cite{Willems} (see also \cite{Maric Yates Kramer 2007}). A $((2^{nR_{1}%
},2^{nR_{2}}),n,\bar{K},K)$ code for the CME\ channel consists of $2\bar{K}$
``conferencing'' functions at the encoders, where $\bar{K}$ is the number of
conferencing rounds between the transmitters $(k=1,2,...,\bar{K})$:
\begin{subequations}
\begin{align}
\bar{h}_{1,k}  &  \text{:}\text{ }\mathcal{W}_{1}\times\mathcal{\bar{V}}%
_{2,1}\times\cdots\times\mathcal{\bar{V}}_{2,k-1}\rightarrow\mathcal{\bar{V}%
}_{1,k}\\
\bar{h}_{2,k}  &  \text{:}\text{ }\mathcal{W}_{2}\times\mathcal{\bar{V}}%
_{1,1}\times\cdots\times\mathcal{\bar{V}}_{1,k-1}\rightarrow\mathcal{\bar{V}%
}_{2,k}.
\end{align}
with alphabets $\mathcal{\bar{V}}_{i,k}$ ($k=1,2,...,\bar{K}$) satisfying the
capacity budget on the conferencing links:
\end{subequations}
\begin{equation}%
{\displaystyle\sum\limits_{k=1}^{K}}
|\mathcal{\bar{V}}_{1,k}|\leq n\bar{C}_{12}\text{ and }%
{\displaystyle\sum\limits_{k=1}^{K}}
|\mathcal{\bar{V}}_{2,k}|\leq n\bar{C}_{21},
\end{equation}
and encoding functions:
\begin{subequations}
\begin{align}
f_{1}  &  \text{:}\text{ }\mathcal{W}_{1}\times\mathcal{\bar{V}}_{2}^{\bar{K}%
}\rightarrow\mathcal{X}_{1}^{n}\\
f_{2}  &  \text{:}\text{ }\mathcal{W}_{2}\times\mathcal{\bar{V}}_{1}^{\bar{K}%
}\rightarrow\mathcal{X}_{2}^{n}.\text{ }%
\end{align}
It is noted that encoding takes place after the $\bar{K}$ conferencing rounds
at the transmit side, similar to the operation at the receivers where decoding
occurs after the $K$ decoder-side conferencing rounds. Decoding and
conferencing at the receiver side are defined as in Sec. \ref{sec_model} (by
setting the common message $W_{0}$ to a constant). Achievability of a rate
pair ($R_{1},R_{2}$)\ is defined by requiring the existence of a code with
such rates and with a vanishing probability of error on the two messages
$W_{1}$ and $W_{2}$. The capacity region of the CME channel is denoted as
$\mathcal{C}_{CME}(\bar{C}_{12}^{{}},\bar{C}_{21},C_{12},C_{21}).$

An outer bound can be established similarly to Proposition \ref{prop1}.
\end{subequations}
\begin{prop}
\label{prop12} We have $\mathcal{C}_{CME}(\bar{C}_{12},\bar{C}_{21}%
,C_{12},C_{21})\subseteq\mathcal{C}_{CME-out}(\bar{C}_{12},\bar{C}_{21}%
,C_{12},C_{21})$ with
\begin{align}
\mathcal{C}_{CME-out}(\bar{C}_{12},\bar{C}_{21},C_{12},C_{21}) =\{(R_{1}%
,R_{2})\text{: }  &  ((R_{12}+R_{21}),R_{1}-R_{12},R_{2}-R_{21})\nonumber\\
&  \in\mathcal{C}_{CM-out}(C_{12},C_{21})\text{ where }R_{12}=\min\{R_{1}%
,\bar{C}_{12}\}\nonumber\\
&  \text{ and }R_{21} =\min\{R_{2},\bar{C}_{21}\}\},
\end{align}
where $\mathcal{C}_{CM-out}(C_{12},C_{21})$ is defined in (\ref{outer bound}).
It is shown in \cite{Maric Yates Kramer 2007} that with only conferencing
encoders we have $\mathcal{C}_{CME}(\bar{C}_{12},\bar{C}_{21},0,0)=\mathcal{C}%
_{CME-out}(\bar{C}_{12},\bar{C}_{21},0,0).$
\end{prop}

\begin{proof}
See Appendix \ref{app:12}.
\end{proof}

The following capacity results can be established similarly to Proposition
\ref{prop2} and \ref{prop3}, respectively.

\begin{prop}
\label{prop13} If the CME\ channel is physically degraded such that
$(X_{1}X_{2})-Y_{1}-Y_{2}$ forms a Markov chain, then the capacity region is
obtained as
\begin{equation}
\mathcal{C}_{CME-DEG}(\bar{C}_{12},\bar{C}_{21},C_{12},C_{21})=\mathcal{C}%
_{CME-out}(\bar{C}_{12},\bar{C}_{21},C_{12},0).
\end{equation}
Notice that here $p^{\ast}(y_{1}y_{2}|x_{1},x_{2})=p(y_{1}|x_{1},x_{2}%
)p(y_{2}|y_{1})$ due to degradedness. A symmetric result holds for the
physically degraded channel $(X_{1}X_{2})-Y_{2}-Y_{1}.$
\end{prop}

\begin{proof}
The converse follows from the same reasoning used in Proposition \ref{prop2}
and Proposition \ref{prop8}. Achievability is obtained by using a scheme
similar to Proposition \ref{prop2} with the only difference being that here
transmission is performed according to the optimal strategy for a MAC\ with
conferencing encoders \cite{Willems} (see also Theorem 2 in \cite{Maric Yates
Kramer 2007}). It is noted that this strategy requires only one conferencing
round at the encoders, $\bar{K}=1$.
\end{proof}

\begin{prop}
\label{prop14} In the case of unidirectional cooperation at the receiver side
($C_{12}=0$ or $C_{21}=0$), the capacity region is given by, respectively,
\begin{subequations}
\begin{align}
\mathcal{C}_{CME}(\bar{C}_{12},\bar{C}_{21},0,C_{21})  &  =\mathcal{C}%
_{CME-out}(\bar{C}_{12},\bar{C}_{21},0,C_{21})
\end{align}
or
\begin{align}
\mathcal{C}_{CME}(\bar{C}_{12},\bar{C}_{21},C_{12},0)  &  =\mathcal{C}%
_{CME-out}(\bar{C}_{12},\bar{C}_{21},C_{12},0).
\end{align}

\end{subequations}
\end{prop}

\textit{Proof}: The proof is similar to those of Proposition \ref{prop3} and
Proposition \ref{prop13}.

It is finally noted that the outer bound and achievable rates derived in Sec.
\ref{sec_achievability} and Section \ref{sec_Gaussian} can also be extended to
the CME\ channel and the Gaussian CME\ channel (\ref{GCM}) following the same
approach used to derive Propositions \ref{prop13} and Proposition
\ref{prop14}, that is, by considering the optimal coding strategy for the MAC
with conferencing encoders \cite{Willems} (which requires $\bar{K}=1$). In
terms of the rate regions, this simply amounts to using the same
transformation from $(R_{0},R_{1},R_{2})$ to ($R_{1},R_{2}$) discussed above
(see also \cite{Maric Yates Kramer 2007}). For instance, an outer bound on the
Gaussian capacity region $\mathcal{C}_{CME}^{\mathcal{G}}(\bar{C}_{12},\bar
{C}_{21},C_{12},C_{21})$ can be obtained as
\begin{align}
\mathcal{C}_{CME-out}^{\mathcal{G}}(\bar{C}_{12},\bar{C}_{21},C_{12},C_{21})
=\{(R_{1},R_{2})\text{: }  &  ((R_{12}+R_{21}),R_{1}-R_{12},R_{2}%
-R_{21})\nonumber\\
&  \in\mathcal{C}_{CM-out}^{\mathcal{G}}(C_{12},C_{21})\text{ where }%
R_{12}=\min\{R_{1},\bar{C}_{12}\}\nonumber\\
&  \text{and }R_{21} =\min\{R_{2},\bar{C}_{21}\}\}, \label{CME G 1}%
\end{align}
and similarly for the rate regions achievable with the one-round and two-round
receiver-side conferencing strategies ((\ref{one step Gaussian}) and
(\ref{two step Gaussian})) coupled with the optimal transmit cooperation
\cite{Willems}.

\begin{rem}
(\textit{Conferencing encoders vs. conferencing decoders}) While no general
capacity results have been derived that enable a conclusive comparison between
the performance of conferencing encoders or decoders in the compound multiple
access channel, some basic conclusions can be drawn based on the analysis
above. To start with, conferencing decoders tend to behave like a
multi-antenna receiver for large conferencing capacities and thus, as
discussed in Sec. \ref{sec_Gaussian}, have the potential for increasing the
multiplexing gain of the sum-rate up to the maximum value of two. In contrast,
it can be seen from the outer bound (\ref{CME G 1}) that conferencing at the
encoders alone does not have such a potential advantage, as the coherent power
combining afforded by cooperating encoders is not enough to increase the
multiplexing gain of the system\footnote{It is noted that this conclusion
would be significantly different for an interference channel, since in this
case conferencing at the encoders has the capability of creating an equivalent
two-antenna broadcast channel with single-antenna receivers, whose
multiplexing gain is known to be two.}. However, this does not necessarily
mean that decoder conferencing is always to be preferred to encoder
conferencing. Consider for instance the case of unidirectional links, where
say $\bar{C}_{21}= C_{21}=0,$ so that conferencing links exist only from
encoder 1 to encoder 2 on the transmit side and from decoder 1 to decoder 2 on
the receive side. In this case, the capacity region is given in Proposition
\ref{prop14}, and one can see that, e.g., for a symmetric system ($\gamma
_{11}^{2}=\gamma_{22}^{2}$, $\gamma_{21}^{2}=\gamma_{12}^{2}$ and $P_{1}%
=P_{2}$), the conferencing link at the decoders alone never helps increase the
achievable rates, while the conferencing link at the transmit side can always
enlarge the achievable rate region. Further performance comparison is carried
out numerically below.
\end{rem}

%

\begin{figure}
\begin{center}
\includegraphics[width=4.0421in]%
{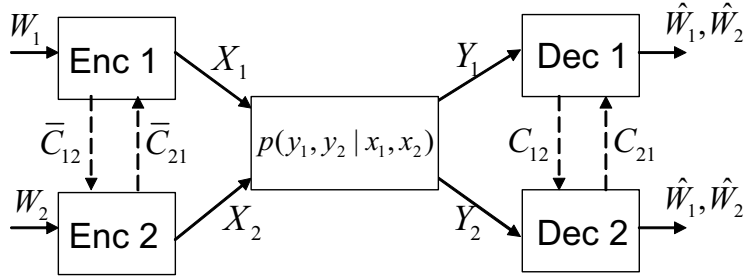}%
\caption{A discrete-memoryless compound MAC channel with conferencing decoders
and encoders (for short, CME).}%
\label{model1}%
\end{center}
\end{figure}

\subsection{Numerical results}

In this section, we present a numerical example related to the scenario in
Fig. \ref{model1} for the Gaussian CME channel (\ref{GCM}). Fig.
\ref{region 2} shows the outer bound (\ref{CME G 1}) evaluated for
encoder-side ($\bar{C}_{12}=\bar{C}_{21}=0$), decoder-side ($C_{12}=C_{21}=0$)
or both-side conferencing, along with the rate regions achievable with
one-round and two-round strategies and with no cooperation for $P_{1}%
=P_{2}=5dB,$ $\gamma_{12}^{2}=\gamma_{21}^{2}=-3dB,$ $\gamma_{11}^{2}%
=\gamma_{22}^{2}=0dB,$ and conferencing capacities (when non-zero) $\bar
{C}_{21}=\bar{C}_{12}=C_{21}=C_{12}=0.3$. Considering first the outer bounds,
it can be seen that both conferencing at the encoders and decoders have the
same potential in terms of increasing the rates $R_{1}$ and $R_{2},$ whereas
for this example the outer bound corresponding to decoder-side cooperation
leads to a larger sum-rate $R_{1}+R_{2}.$ Comparison of achievable rates via
one or two rounds of conferencing at the receiver side (recall that one round
of encoder conferencing is enough to achieve all the rate points discussed
here) is similar to that seen in Fig. \ref{fig1}.%

\begin{figure}
\begin{center}
\includegraphics[width=4.2384in]%
{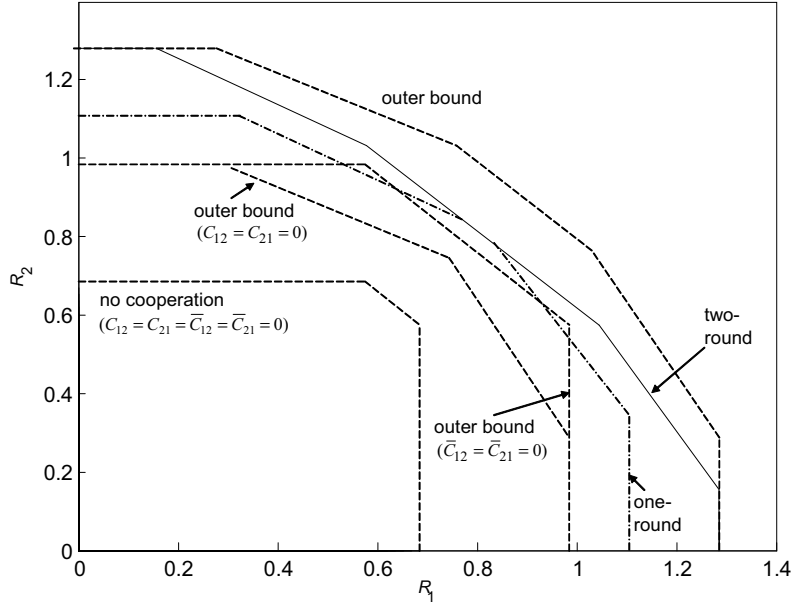}%
\caption{Outer bound (\ref{CME G 1}) evaluated for encoder-side ($\bar{C}%
_{12}=\bar{C}_{21}=0$), decoder-side ($C_{12}=C_{21}=0)$ or both-side
conferencing, along with the rate regions achievable with one-round and
two-round strategies and with no cooperation for $P_{1}=P_{2}=5dB,$
$\gamma_{12}^{2}=\gamma_{21}^{2}=-3dB,$ $\gamma_{11}^{2}=\gamma_{22}^{2}=0dB,$
and conferencing capacities (when non-zero) $\bar{C}_{21}=\bar{C}_{12}%
=C_{21}=C_{12}=0.3$. }%
\label{region 2}%
\end{center}
\end{figure}

\section{Conclusions}

The model of conferencing encoders and/or decoders is a convenient framework
that allows evaluation of the potential gains arising from cooperation at the
transmitter or receiver side in a wireless network. From a practical
standpoint, it accounts for scenarios where out-of-band signal paths exist at
the two ends of a communication link, as is the case in wireless communication
systems where nodes are endowed with multiple radio interfaces. In this work,
we have contributed to the state of knowledge in this area by investigating a
compound MAC with conferencing decoders and, possibly, encoders. The compound
MAC can be seen as a combination of two single-message broadcast (multicast)
channels from the standpoint of the transmitters, or two multiple access
channels as seen by the receivers. The scenario at hand generalizes a number
of previously studied setups, such as MAC\ or compound MAC with common message
or conferencing encoders and single-message broadcast channel with two
conferencing decoders. A number of capacity results have been derived that
have shed light on the impact of decoder and encoder conferencing on the
capacity of the compound MAC. Among the conclusions, we have shown that in a
compound Gaussian MAC, one round of conferencing at the decoders achieves the
entire capacity region within a constant number of bits/s/Hz in several
special cases. One round of conferencing at the transmitters is also optimal
in all the cases where the capacity region is known. Moreover, comparing the
performance of conferencing at the encoders and decoders, it has been pointed
out that examples can be constructed where either one outperforms the other.
However, in the Gaussian case, while conferencing at the decoders has the
potential of increasing the sum-rate multiplexing gain to the optimal value of
two by mimicking a multiantenna receiver, the same is not true of conferencing
encoders, since coherent power combining afforded by cooperating encoders is
not enough to increase the multiplexing gain beyond one (recall that the two
decoders must estimate both messages).

As a possible extension of this work we mention the study of an interference
channel, rather than the compound MAC, with conferencing decoders. As already
pointed out in the paper, some of the conclusions here would be significantly
different in this case, and the analysis could benefit from the techniques
used in \cite{Avestimehr} \cite{Avestimehr 1} to study interference channels
with no cooperation.

\appendices

\section{Proof of Proposition \ref{prop1}}

\label{app:1}

In order for rates ($R_{0},R_{1},R_{2}$) to be achievable, the probability of
error $P_{e}$ needs to satisfy (\ref{Pe1}) which, by the union bound, is
implied by $P_{e,i}\leq\varepsilon/2$ for $i=1,2$ with \[P_{e,1}=\frac
{1}{2^{n(R_{0}+R_{1}+R_{2})}}\sum_{\bar{w}\in\mathcal{W}_{0}\times
\mathcal{W}_{1}\times\mathcal{W}_{2}}\Pr[h_{1}(Y_{1}^{n},V_{2}^{k})\neq\bar
{w}|\bar{w}\mbox{ sent}] \] and similarly for $P_{e,2}.$ Consider the first
receiver. By Fano's inequality, we have
\begin{equation}
H(W_{0},W_{1},W_{2}|Y_{1}^{n},V_{2}^{K})\leq H(P_{e,1})+n(R_{0}+R_{1}%
+R_{2})P_{e,1}\triangleq n\delta_{n} \label{ineq1}%
\end{equation}
with $\delta_{n}\rightarrow0$ as $n\rightarrow\infty.$ It also follows that
\begin{subequations}
\label{ineq3}%
\begin{align}
H(W_{1}W_{2}|Y_{1}^{n},V_{2}^{K},W_{0})  &  \leq n\delta_{n},\label{ineq2}\\
H(W_{1}|Y_{1}^{n},V_{2}^{K},W_{0},W_{2})  &  \leq n\delta_{n} \mbox{ and }\\
H(W_{2}|Y_{1}^{n},V_{2}^{K},W_{0},W_{1})  &  \leq n\delta_{n}.\text{ }%
\end{align}
Now, from (\ref{ineq1}), we have
\end{subequations}
\begin{align*}
n(R_{0}+R_{1}+R_{2})  &  \leq I(W_{0},W_{1},W_{2};Y_{1}^{n},V_{2}^{K}%
)+n\delta_{n}\\
&  \leq I(W_{0},W_{1},W_{2};Y_{1}^{n})+I(W_{0},W_{1},W_{2};V_{2}^{K}|Y_{1}%
^{n})+n\delta_{n}\\
&  \overset{(a)}{\leq}I(W_{0},W_{1},W_{2};Y_{1}^{n})+nC_{21}+n\delta_{n}\\
&  \overset{(b)}{\leq}^{{}}%
{\displaystyle\sum\limits_{i=1}^{n}}
I(X_{1,i},X_{2,i};Y_{1,i})+nC_{21}+n\delta_{n},
\end{align*}
where (a) follows from the fact that $I(W_{0},W_{1},W_{2};V_{2}^{K}|Y_{1}%
^{n})\leq H(V_{2}^{K})\leq nC_{21}$ and (b) is obtained similarly to
\cite{Willems}, Sec. 3.4$.$ From (\ref{ineq3}), using similar arguments as in
the above chain of inequalities, one can also obtain
\begin{align*}
n(R_{1}+R_{2})  &  \leq%
{\displaystyle\sum\limits_{i=1}^{n}}
I(X_{1,i},X_{2,i};Y_{1,i}|W_{0})+nC_{21}+n\delta_{n},\\
nR_{1}  &  \leq%
{\displaystyle\sum\limits_{i=1}^{n}}
I(X_{1,i};Y_{1,i}|X_{2,i},W_{0})+nC_{21}+n\delta_{n} \mbox{ and }\\
nR_{2}  &  \leq%
{\displaystyle\sum\limits_{i=1}^{n}}
I(X_{2,i};Y_{1,i}|X_{1,i},W_{0})+nC_{21}+n\delta_{n}.
\end{align*}
Now defining $U_{i}=W_{0}$, the proof is completed as in \cite{Maric Yates
Kramer 2007}. We can repeat the same arguments for receiver 2. Also the
condition that $(R_{0},R_{1},R_{2})\in\mathcal{R}_{MAC,FC}$ follows similarly
considering full cooperation between the receivers.

\section{Proof of Proposition \ref{prop2}}

\label{app:2}

\textit{Converse}: The converse follows immediately from Proposition
\ref{prop1} and the data processing theorem. In fact, it is easy to see that,
because of physical degradedness, receiver 1 cannot benefit from $V_{2}^{K},$
which is a function of $Y_{2}^{n}$ and $Y_{1}^{n}$ via $V_{1}^{k}.$ For
instance, condition (\ref{ineq1}) now becomes%
\[
H(W_{0},W_{1},W_{2}|Y_{1}^{n})=H(W_{0},W_{1},W_{2}|Y_{1}^{n},V_{2}^{K})\leq
H(P_{e,1})+n(R_{1}+R_{2})P_{e,1}\triangleq n\delta_{n},
\]
due to the Markov chain $(W_{0},W_{1},W_{2})-Y_{1}^{n}-V_{2}^{K}$. Repeating
the same arguments for the other conditions (\ref{ineq3}), the converse is
then completed as in Proposition \ref{prop1}.

\textit{Achievability}: Codeword generation at the transmitters is performed
as for the MAC with common information \cite{Willems} \cite{Slepian Wolf}:

Generate $2^{nR_{0}}$ sequences $u^{n}(w_{0})$ of length $n$, with the
elements of each being chosen independent identically distributed (i.i.d.)
according to the distribution $p(u),$ $w_{0}\in\mathcal{W}_{0}.$ For any
sequence $u^{n}(w_{0}),$ generate $2^{nR_{i}}$ independent sequences
$x_{i}^{n}(w_{0},w_{i}),$ $w_{i}\in\mathcal{W}_{i},$ again i.i.d. according to
$p(x_{i}|u_{i}(w_{0})),$ for $i=1,2$.

At receiver 1, the message sets $\mathcal{W}_{0},$ $\mathcal{W}_{1}$ and
$\mathcal{W}_{2}$ are partitioned into $2^{n\alpha_{0}C_{12}}$, $2^{n\alpha
_{1}C_{12}}$ and $2^{n\alpha_{2}C_{12}}$ subsets$,$ respectively, for given
$0\leq\alpha_{i}\leq1$ and $\sum_{i=0}^{2}\alpha_{i}=1$.\ This is done by
assigning each codeword in the message sets $\mathcal{W}_{0},$ $\mathcal{W}%
_{1}$ and $\mathcal{W}_{2}$ independently and randomly to the index sets
$\{1,2,...,2^{n\alpha_{0}C_{12}}\}$, $\{1,2,...,2^{n\alpha_{1}C_{12}}\}$ and
$\{1,2,...,2^{n\alpha_{2}C_{12}}\},$\ respectively.

\textit{Encoding} at transmitter $i$ is performed by sending codeword
$x_{i}^{n}(w_{0},w_{i})$ corresponding to the common message $w_{0}%
\in\mathcal{W}_{0}$ and local message $w_{i}\in\mathcal{W}_{i}$ $(i=1,2).$
Encoding at decoder 1 takes place after detection of the two messages $W_{0},
$ $W_{1}$ and $W_{2}$ (see description of decoding below). In particular,
decoder 1 sends over the conferencing link 1-2 the indices of the subsets in
which the estimated messages $W_{0},$ $W_{1}$ and $W_{2}$ lie. Notice that
this requires $nC_{12}$ bits and $K=1$ conferencing rounds (i.e.,
$|\mathcal{V}_{1,1}|=nC_{12}$). Also we emphasize again that the conferencing
link 2-1 is not used ($|\mathcal{V}_{2,k}|=0$).

\textit{Decoding} at the first decoder is carried out by finding jointly
typical sequences $(u^{n}(w_{0}),$ $x_{1}^{n}(w_{0},w_{1}),$ $x_{2}^{n}%
(w_{0},w_{2}),$ $y_{1}^{n})$ with $w_{i}\in\mathcal{W}_{i}$ \cite{Cover}$.$ As
discussed above, once the first decoder has obtained the messages $W_{0},$
$W_{1}$ and $W_{2}$, it sends the corresponding subset indices to receiver 2
over the conferencing channels. Decoding at receiver 2 then takes place again
based on a standard MAC joint-typicality encoder with the caveat that the
messages $W_{0},$ $W_{1}$ and $W_{2}$ are now known to belong to the reduced
set given by the subsets mentioned above.

The\textit{\ analysis of the probability of error} follows immediately from
\cite{Willems} \cite{Slepian Wolf}. In particular, as far as receiver 1 is
concerned, it can be seen from \cite{Willems} \cite{Slepian Wolf} that a
sufficient condition for the probability of error to approach zero as
$n\rightarrow\infty$ is given by $(R_{0},R_{1},R_{2})\in\mathcal{R}%
_{MAC,1}(p(u),p(x_{1}|u),p(x_{2}|u)).$ Considering receiver 2, a sufficient
condition is that the rates belong to the region
\begin{subequations}
\label{alphas}%
\begin{align}
\{(R_{0},R_{1},R_{2})  &  \text{:}\text{ }R_{0}\geq0,\text{ }R_{1}\geq0,\text{
}R_{2}\geq0,\\
R_{1}  &  \leq I(X_{1};Y_{2}|X_{2}U)+\alpha_{1}C_{12}\\
R_{2}  &  \leq I(X_{2};Y_{2}|X_{1}U)+\alpha_{2}C_{12}\\
R_{1}+R_{2}  &  \leq I(X_{1}X_{2};Y_{2}|U)+(\alpha_{1}+\alpha_{2})C_{12}\\
R_{0}+R_{1}+R_{2}  &  \leq I(X_{1}X_{2};Y_{2})+C_{12}\},
\end{align}
for the given $\alpha_{i}.$ Taking the union over all allowed $\alpha_{i}$ in
(\ref{alphas}) concludes the proof.

\section{Proof of Proposition \ref{prop9}}

\label{app:9}

We first prove that $R_{1,OR}\geq R_{1,out}-\frac{1}{2}.$ We consider three
separate cases and show that the statement of the theorem holds for each case
separately. We define $P_{a}\triangleq\gamma_{11}^{2}P_{1}$, $P_{b}%
\triangleq\gamma_{12}^{2}P_{1},$ $\breve{C}_{12}\triangleq2^{2C_{12}}-1$ and
$\breve{C}_{21}\triangleq2^{2C_{21}}-1$ for simplicity of notation. We remark
that using this notation the compression noises (\ref{sigmas1}) can be written
for the case at hand as $\sigma_{1}^{2}=\frac{1+P_{a}+P_{b}}{(1+P_{b}%
)\breve{C}_{12}}$ and $\sigma_{2}^{2}-\frac{1+P_{a}+P_{b}}{(1+P_{a})\tilde
{P}_{21}}.$

\textit{Case 1}: Let
\end{subequations}
\begin{align}
\breve{C}_{21}  &  \geq\frac{P_{b}}{1+P_{a}} \label{c1_c1}%
\end{align}
and
\begin{align}
\breve{C}_{12}  &  \geq\frac{P_{a}}{1+P_{b}}. \label{c1_c2}%
\end{align}
In this case, the upper bound (\ref{R1out}) is $R_{1,out}=\frac{1}{2}%
\log(1+P_{a}+P_{b})$ and for the achievable rate with one-round conferencing
(\ref{R1or}) we have
\begin{align}
\mathcal{C}\left(  P_{a}+\frac{P_{b}}{1+\sigma_{2}^{2}}\right)   &  =\frac
{1}{2}\log\left(  1+P_{a}+\frac{P_{b}}{1+\frac{1+P_{a}+P_{b}}{\tilde{P}%
_{21}(1+P_{a})}}\right) \nonumber\\
&  \geq\frac{1}{2}\log\left(  1+P_{a}+\frac{P_{b}^{2}}{1+P_{a}+2P_{b}}\right)
\label{c1_e1}
\end{align}
\begin{align}
&  =\frac{1}{2}\log\left(  \frac{(1+P_{a}+P_{b})^{2}}{1+P_{a}+2P_{b}}\right)
\nonumber\\
&  =\frac{1}{2}\log\left(  1+P_{a}+P_{b}\right)  +\frac{1}{2}\log\left(
\frac{1+P_{a}+P_{b}}{1+P_{a}+2P_{b}}\right) \nonumber\\
&  \geq R_{out}-\frac{1}{2}.\nonumber
\end{align}
where (\ref{c1_e1}) follows from (\ref{c1_c1}). Similarly, using
(\ref{c1_c2}), we can also show that $\mathcal{C}\left(  P_{b}+\frac{P_{a}%
}{1+\sigma_{1}^{2}}\right)  \geq R_{out}-\frac{1}{2}.$ It then follows that
$R_{1,out}\geq R_{out}-\frac{1}{2}.$

\textit{Case 2}: Now,\emph{\ }let
\begin{align}
\breve{C}_{21}  &  \leq\frac{P_{b}}{1+P_{a}} \label{c2_c1}%
\end{align}
and
\begin{align}
(1+P_{a})(1+\breve{C}_{21})  &  \leq(1+P_{b})(1+\breve{C}_{12}). \label{c2_c2}%
\end{align}
In this case, we have $R_{1,out}=\frac{1}{2}\log(1+P_{a})(1+\breve{C}_{21})$
and
\begin{align}
\mathcal{C}\left(  P_{a}+\frac{P_{b}}{1+\sigma_{2}^{2}}\right)   &  =\frac
{1}{2}\log\left(  1+P_{a}+\frac{P_{b}}{1+\frac{1+P_{a}+P_{b}}{\tilde{P}%
_{21}(1+P_{a})}}\right) \nonumber\\
&  =\frac{1}{2}\log\left(  1+P_{a}\right)  +\frac{1}{2}\log\left(
1+\frac{P_{b}\breve{C}_{21}}{(1+P_{a})\breve{C}_{21}+(1+P_{a}+P_{b})}\right)
\nonumber\\
&  =\frac{1}{2}\log\left(  1+P_{a}\right)  +\frac{1}{2}\log\left(
\frac{(1+P_{a}+P_{b})(1+\breve{C}_{21})}{(1+P_{a})\breve{C}_{21}%
+(1+P_{a}+P_{b})}\right) \nonumber\\
&  =\frac{1}{2}\log\left(  1+P_{a}\right)  \left(  1+\breve{C}_{21}\right)
+\frac{1}{2}\log\left(  \frac{1+P_{a}+P_{b}}{(1+P_{a})\breve{C}_{21}%
+(1+P_{a}+P_{b})}\right) \nonumber\\
&  \geq R_{out}+\frac{1}{2}\log\left(  \frac{1+P_{a}+P_{b}}{1+P_{a}+2P_{b}%
}\right) \label{c2_e1}\\
&  \geq R_{out}-\frac{1}{2}, \label{c2_e1b}%
\end{align}
where (\ref{c2_e1}) follows from (\ref{c2_c1}). On the other hand, we also
have
\begin{align}
\mathcal{C}\left(  P_{b}+\frac{P_{a}}{1+\sigma_{1}^{2}}\right)   &  =\frac
{1}{2}\log\left(  1+P_{b}+\frac{P_{a}}{1+\frac{1+P_{a}+P_{b}}{\tilde{P}%
_{12}(1+P_{b})}}\right) \nonumber\\
&  =\frac{1}{2}\log\left[  (1+P_{b})\left(  1+\frac{P_{a}}{1+P_{b}%
+\frac{1+P_{a}+P_{b}}{\breve{C}_{21}}}\right)  \right] \nonumber
\end{align}
\begin{align}
&  \geq\frac{1}{2}\log\left[  (1+P_{b})\left(  1+\frac{P_{a}[(1+P_{a}%
)(1+\breve{C}_{21})-(1+P_{b})]}{(1+P_{b})[(1+P_{a})(1+\breve{C}_{21})+P_{a}%
]}\right)  \right] \label{c2_e2}\\
&  =\frac{1}{2}\log\left[  \frac{(1+P_{a})(1+\breve{C}_{21})(1+P_{a}+P_{b}%
)}{(1+P_{a})(1+\breve{C}_{21})+P_{a}}\right] \nonumber\\
&  =R_{out}+\frac{1}{2}\log\left[  \frac{1+P_{a}+P_{b}}{(1+P_{a})(1+\tilde
{P}_{21})+P_{a}}\right] \nonumber\\
&  \geq R_{out}+\frac{1}{2}\log\left[  \frac{1+P_{a}+P_{b}}{1+2P_{a}+P_{b}%
}\right] \label{c2_e3}\\
&  \geq R_{out}-\frac{1}{2}, \label{c2_e4}%
\end{align}
where (\ref{c2_e2}) follows from (\ref{c2_c2}); and (\ref{c2_e3}) follows from
(\ref{c2_c1}). From (\ref{c2_e1b}) and (\ref{c2_e4}), we see that the theorem
holds for Case 2 as well.

\textit{Case 3}: Let
\begin{align}
\breve{C}_{12}  &  \leq\frac{P_{a}}{1+P_{b}} \label{c3_c1}%
\end{align}
and
\begin{align}
(1+P_{b})(1+\breve{C}_{12})  &  \leq(1+P_{a})(1+\breve{C}_{21}). \label{c3_c2}%
\end{align}
In this case, $R_{out}=\frac{1}{2}\log(1+P_{b})(1+\breve{C}_{12}).$ Case 3
follows similarly to Case 2.

Now, for the symmetric channel case, i.e., $\gamma_{11}^{2}=\gamma_{12}^{2} $,
that is, if $P_{a}=P_{b}\triangleq P$, we have to prove that $R_{OR}\geq
R_{out}-0.29.$ This follows similarly to the treatment above as%
\begin{align}
R_{OR}  &  \geq R_{out}-\frac{1}{2}\log\left[  \frac{1+P_{a}+P_{b}}%
{1+2P_{a}+P_{b}}\right] \nonumber\\
&  =R_{out}-\frac{1}{2}\log\left[  \frac{1+2P}{1+3P}\right] \nonumber\\
&  \geq R_{out}-\frac{1}{2}(\log3-1).
\end{align}

\section{Proof of Proposition \ref{prop10}}

\label{app:10}

To prove the theorem, we show that the bounds for $R_{1}$, $R_{2}$ and
$R_{1}+R_{2}$ in $\mathcal{C}_{OR}^{\mathcal{G}}(C)$ are all within
$\frac{\log(1+\beta)}{2}$ bits of the corresponding bounds in $\mathcal{C}%
_{CM-out}^{\mathcal{G}}(C)$. We define $\breve{C}\triangleq2^{2C}-1$. Without
loss of generality, we assume $b\geq a$, and define $x\triangleq aP$. Then
from the definition of $\beta$, we get $bP=\beta x$. The outer bound and the
achievable rates with one-round conferencing can now be written as
\begin{subequations}
\begin{align}
\mathcal{C}_{CM-out}^{\mathcal{G}}(C)  &  =\{(R_{1},R_{2}):R_{1}\geq
0,\text{{}}R_{2}\geq0,\nonumber\\
R_{1}  &  \leq\min\{\mathcal{C}(x)+C,~\mathcal{C}((1+\beta)x)\},\\
R_{2}  &  \leq\min\{\mathcal{C}(x)+C,~\mathcal{C}((1+\beta)x)\},\\
R_{1}+R_{2}  &  \leq\min\{\mathcal{C}((1+\beta)x)+C,\text{ }\mathcal{C}%
(2(1+\beta)x+(\beta-1)^{2}x^{2})\}\},
\end{align}
and
\end{subequations}
\begin{subequations}
\begin{align}
\mathcal{R}_{OR}^{\mathcal{G}}(C)  &  =\left\{  (R_{1},R_{2}):R_{1}%
\geq0,\text{ }R_{2}\geq0,\right. \nonumber\\
R_{1}  &  \leq\mathcal{C}\left(  \left(  1+\frac{\beta}{1+\sigma^{2}}\right)
x\right)  ,\\
R_{2}  &  \leq\mathcal{C}\left(  \left(  1+\frac{\beta}{1+\sigma^{2}}\right)
x\right)  ,\\
&  \left.  R_{1}+R_{2}\leq\mathcal{C}\left(  \left(  1+\beta+\frac{1+\beta
}{1+\sigma^{2}}+\frac{(\beta-1)^{2}x}{1+\sigma^{2}}\right)  x\right)
\right\}
\end{align}
with $\sigma^{2}\triangleq\frac{1+2(1+\beta)x+(\beta-1)^{2}x^{2}}%
{(1+(1+\beta)x)\breve{C}}$, respectively.

We first define functions $A$ and $B$ as
\end{subequations}
\begin{align}
A(x)  &  \triangleq1+(1+\beta)x\nonumber
\end{align}
and
\begin{align}
B(x)  &  \triangleq1+2(1+\beta)x+(\beta-1)^{2}x^{2}.
\end{align}

Consider the bound on $R_{1}$. We analyze two cases separately. If $\breve
{C}\geq\frac{\beta x}{1+x}$, then the outer bound is equivalent to $R_{1}%
\leq\frac{1}{2}\log A$. On the other hand, the bound on the achievable $R_{1}$
is found as
\begin{align}
\frac{1}{2}\log\left(  1+x+\frac{\beta x}{1+\frac{1+2(1+\beta)x+(\beta
-1)^{2}x^{2}}{(1+(1+\beta)x)\breve{C}}}\right)   &  =\frac{1}{2}\log\left(
1+x+\frac{\beta x}{1+\frac{B}{A\breve{C}}}\right) \nonumber\\
&  \geq\frac{1}{2}\log\left(  1+x+\frac{\beta^{2}Ax^{2}}{\beta Ax+(1+x)B}%
\right) \nonumber\\
&  \geq\frac{1}{2}\log\left[  A\frac{(\beta^{2}+\beta)x^{2}+\beta
x+(1+x)^{2}\left(  1+\frac{(\beta-1)^{2}}{\beta+1}x\right)  }{\beta
Ax+(1+x)B}\right] \nonumber
\end{align}
\begin{align}
&  \geq\frac{1}{2}\log\frac{A}{1+\beta}\nonumber\\
&  =\frac{1}{2}\log A-\frac{1}{2}\log(1+\beta). \label{cMAC_b1}%
\end{align}
If $\breve{C}\leq\frac{\beta x}{1+x}$, then the outer bound is equivalent to
$R_{1}\leq\frac{1}{2}\log((1+\breve{C})(1+x))$. The achievable rate bound can
be written as
\begin{align}
\frac{1}{2}\log\left(  1+x+\frac{\beta x}{1+\frac{B}{A\breve{C}}}\right)   &
=\frac{1}{2}\log\left(  \frac{(1+x)(A\breve{C}+B)+\beta A\breve{C}x}%
{A\breve{C}+B}\right) \nonumber\\
&  \geq\frac{1}{2}\log(1+x)\left(  \frac{(1+x)B+A^{2}\breve{C}}{\beta
xA+(1+x)B}\right) \nonumber\\
&  \geq\frac{1}{2}\log(1+x)\frac{1+\breve{C}}{1+\beta}\nonumber\\
&  =\frac{1}{2}\log(1+x)(1+\breve{C})-\frac{1}{2}\log(1+\beta).
\label{cMAC_b2}%
\end{align}
Combining (\ref{cMAC_b1}) and (\ref{cMAC_b2}), we conclude that the difference
between the achievable rate bound and the outer bound on $R_{1}$ is not more
than $\frac{1}{2}\log(1+\beta)$ bits. The same result applies for the bounds
on $R_{2}$ in the same way.

Next, we consider the bounds on the sum-rate. If $\breve{C}\geq\frac{B-A}{A}$,
then the outer bound on the sum-rate is equivalent to $R_{1}+R_{2}\leq\frac
{1}{2}\log B$. On the other hand, the bound on achievable sum-rate is
\begin{align}
&  \frac{1}{2}\log\left[  1+(1+\beta)x\left(  1+\frac{1}{1+\sigma^{2}}\right)
+\frac{(\beta-1)^{2}x^{2}}{1+\sigma^{2}}\right] \nonumber\\
&  \geq\frac{1}{2}\log\left[  1+(1+\beta)x\left(  1+\frac{1}{2+\frac{A}{B-A}%
}\right)  +\frac{(\beta-1)^{2}x^{2}}{2+\frac{A}{B-A}}\right] \nonumber\\
&  =\frac{1}{2}\log\left(  A+\frac{B-A}{2+\frac{A}{B-A}}\right) \nonumber\\
&  =\frac{1}{2}\log\left(  \frac{B^{2}}{2B-A}\right) \nonumber\\
&  \geq\frac{1}{2}\log B-\frac{1}{2}. \label{cMAC_sb1}%
\end{align}
If $\breve{C}\leq\frac{B-A}{A}$, then the sum-rate outer bound is equivalent
to $R_{1}+R_{2}\leq\frac{1}{2}\log(1+\breve{C})A$. The achievable sum-rate
bound is
\begin{align}
&  \frac{1}{2}\log\left[  1+(1+\beta)x\left(  1+\frac{1}{1+\sigma^{2}}\right)
+\frac{(\beta-1)^{2}x^{2}}{1+\sigma^{2}}\right] \nonumber\\
&  \geq\frac{1}{2}\log\left[  1+(1+\beta)x\left(  \frac{2\breve{C}A+B}%
{\breve{C}A+B}\right)  +\frac{(\beta-1)^{2}x^{2}\breve{C}A}{\breve{C}%
A+B}\right] \nonumber\\
&  =\frac{1}{2}\log\left[  A+\frac{A\breve{C}(B-A)}{A\breve{C}+B}\right]
\nonumber\\
&  =\frac{1}{2}\log\left[  \frac{AB(1+\breve{C})}{A\breve{C}+B}\right]
\nonumber\\
&  \geq\frac{1}{2}\log\left[  A(1+C)\frac{B}{2B-A}\right] \nonumber\\
&  \geq\frac{1}{2}\log A(1+\breve{C})-\frac{1}{2}. \label{cMAC_sb2}%
\end{align}
From (\ref{cMAC_sb1}) and (\ref{cMAC_sb2}), we see that the difference between
the achievable sum-rate bound and the sum-rate outer bound is always within
half a bit. The claim for the case $a=b$ can be similarly proved. This
concludes the proof.

\section{Proof of Proposition \ref{prop12}}

\label{app:12}

In order for rates ($R_{1},R_{2}$) to be achievable, by Fano's inequality, we
have for the first receiver (see also proof of Proposition \ref{prop1}):
\begin{equation}
H(W_{1},W_{2}|Y_{1}^{n},V_{2}^{K})\leq n\delta_{n} \label{cme1}%
\end{equation}
with $\delta_{n}\rightarrow0$ as $n\rightarrow\infty$. From the previous
inequality, it also follows that ($i=1,2$)
\begin{subequations}
\label{cme2}%
\begin{align}
H(W_{1},W_{2}|Y_{1}^{n},\bar{V}_{1}^{\bar{K}},\bar{V}_{2}^{\bar{K}},V_{2}%
^{K})  &  \leq n\delta_{n},\label{first cme}\\
H(W_{1}|Y_{1}^{n},\bar{V}_{1}^{\bar{K}},\bar{V}_{2}^{\bar{K}},V_{2}^{K}%
,W_{2})  &  \leq n\delta_{n} \mbox{ and }\\
H(W_{2}|Y_{1}^{n},\bar{V}_{1}^{\bar{K}},\bar{V}_{2}^{\bar{K}},V_{2}^{K}%
,W_{1})  &  \leq n\delta_{n},\text{ }%
\end{align}
where $\bar{V}_{1}^{\bar{K}},\bar{V}_{2}^{\bar{K}}$ represent the signals
exchanged during the $\bar{K}$ encoder-side conferencing rounds$.$ Now, we can
treat (\ref{cme1})-(\ref{cme2}) similarly to the proof of Proposition
\ref{prop1} and using the approach in \cite{Willems}. For instance, from
(\ref{first cme}), we have
\end{subequations}
\begin{align*}
n(R_{1}+R_{2})  &  \leq I(W_{1},W_{2};Y_{1}^{n},\bar{V}_{1}^{\bar{K}},\bar
{V}_{2}^{\bar{K}},V_{2}^{K})+n\delta_{n}\\
&  \leq I(W_{1},W_{2};Y_{1}^{n}|\bar{V}_{1}^{\bar{K}},\bar{V}_{2}^{\bar{K}%
})+I(W_{1},W_{2};\bar{V}_{1}^{\bar{K}},\bar{V}_{2}^{\bar{K}})\\
&  ~~~~ +I(W_{1},W_{2};V_{2}^{K}|Y_{1}^{n},\bar{V}_{1}^{\bar{K}},\bar{V}%
_{2}^{\bar{K}})+n\delta_{n}\\
&  \overset{(a)}{\leq}I(W_{1},W_{2};Y_{1}^{n}|\bar{V}_{1}^{\bar{K}},\bar
{V}_{2}^{\bar{K}})+n(\bar{C}_{12}+\bar{C}_{21})+nC_{21}+n\delta_{n}\\
&  \overset{(b)}{\leq}%
{\displaystyle\sum\limits_{i=1}^{n}}
I(X_{1,i},X_{2,i};Y_{1,i}|\bar{V}_{1}^{\bar{K}},\bar{V}_{2}^{\bar{K}}%
)+n(\bar{C}_{12}+\bar{C}_{21})+nC_{21}+n\delta_{n},
\end{align*}
where (a) follows from the fact that $I(W_{1},W_{2};\bar{V}_{1}^{\bar{K}}%
,\bar{V}_{2}^{\bar{K}})\leq H(\bar{V}_{1}^{\bar{K}},\bar{V}_{2}^{\bar{K}})\leq
n(\bar{C}_{12}+\bar{C}_{21})$ and $I(W_{1},W_{2};V_{2}^{K}|Y_{1}^{n},\bar
{V}_{1}^{\bar{K}},\bar{V}_{2}^{\bar{K}})\leq H(V_{2}^{K})\leq nC_{21}$, and
(b) is obtained similarly to \cite{Willems}, Sec. 3.4$.$ From (\ref{cme1}) and
the remaining inequalities in (\ref{cme2}), using similar arguments as in the
above chain of inequalities, one obtains, respectively,
\begin{align*}
n(R_{1}+R_{2})  &  \leq%
{\displaystyle\sum\limits_{i=1}^{n}}
I(X_{1,i},X_{2,i};Y_{1,i})+nC_{21}+n\delta_{n},\\
nR_{1}  &  \leq%
{\displaystyle\sum\limits_{i=1}^{n}}
I(X_{1,i};Y_{1,i}|X_{2,i},\bar{V}_{1}^{\bar{K}},\bar{V}_{2}^{\bar{K}}%
)+n\bar{C}_{12}+nC_{21}+n\delta_{n} \mbox{ and }\\
nR_{2}  &  \leq%
{\displaystyle\sum\limits_{i=1}^{n}}
I(X_{2,i};Y_{1,i}|X_{1,i},\bar{V}_{1}^{\bar{K}},\bar{V}_{2}^{\bar{K}}%
)+n\bar{C}_{21}+nC_{21}+n\delta_{n}.
\end{align*}
Now defining $U_{i}=(\bar{V}_{1}^{\bar{K}},\bar{V}_{2}^{\bar{K}})$, the proof
is completed similarly to Proposition \ref{prop1}, and by repeating the same
arguments for receiver 2.


\begin{thebibliography}{99}                                                                                               %


\bibitem {Willems}F. M. J. Willems, \textit{Informationtheoretical Results for
the Discrete Memoryless Multiple Access Channel}, Ph.D. thesis, Katholieke
Universiteit Leuven, Leuven, Belgium, 1982.

\bibitem {Lapidoth}S. I. Bross, A. Lapidoth and M. A. Wigger,
\textquotedblleft The Gaussian MAC\ with conferencing
encoders,\textquotedblright\ in \textit{Proc. IEEE International Symposium
Inform. Theory} (ISIT 2008), Toronto, ON, Canada, July 6-11, 2008.

\bibitem {Maric Yates Kramer 2007}I. Maric, R. Yates, and G. Kramer,
\textquotedblleft Capacity of interference channels with partial transmitter
cooperation,\textquotedblright\ \textit{IEEE Trans. Inform. Theory}, vol. 53,
no. 10, pp. 3536-3548, Oct. 2007.

\bibitem {Ng Maric et al 2006}C. T. K. Ng, I. Maric, A. J. Goldsmith, S.
Shamai, and R. D. Yates, \textquotedblleft Iterative and one-shot conferencing
in relay channels,\textquotedblright\ in \textit{Proc. IEEE Information Theory
Workshop }(ITW 2006), Punta del Este, Uruguay, Mar. 13-17, 2006.

\bibitem {Dabora Servetto 2006}R. Dabora and S. Servetto, \textquotedblleft
Broadcast channels with cooperating decoders,\textquotedblright\ \textit{IEEE
Trans. Inform. Theory}, vol. 52, no.12, pp. 5438-5454, Dec. 2006.

\bibitem {Dabora Servetto 2006 bis}R. Dabora and S. Servetto,
\textquotedblleft A multi-step conference for cooperative
broadcast,\textquotedblright\ in \textit{Proc. IEEE International Symposium
Information Theory} (ISIT 2006), pp. 2190-2194, Seattle, WA, July 9-14, 2006.

\bibitem {Draper et al 2003}S. C. Draper, B. J. Frey and F. R. Kschischang,
\textquotedblleft Interactive decoding of a broadcast
message,\textquotedblright\ \textit{Proc. Forty-First Annual Allerton
Conference on Communication, Control, and Computing}, Monticello, IL, 2003.

\bibitem {Draper et al 2004}S. C. Draper, B. J. Frey and F. R. Kschischang,
\textquotedblleft On interacting encoders and decoders in multiuser
settings,\textquotedblright\ in \textit{Proc. IEEE International Symposium
Information Theory} (ISIT 2004), Chicago, IL, June 27 - July 2, 2004.

\bibitem {Lasaulce}S. Lasaulce and A. G. Klein, \textquotedblleft Gaussian
broadcast channels with cooperating receivers: the single common message
case,,\textquotedblright\ in\ \textit{Proc. IEEE International Conference
Acoustics, Speech and Signal Processing }(ICASSP 2006), Toulouse, France, May 2006.

\bibitem {ETH master}V. Venkatesan, \textquotedblleft Cooperation required
between destinations in a two-source two-destination network to achieve full
multiplexing gain,\textquotedblright\ Semester thesis, ETH\ Zurich,
Switzerland, Winter 2006.

\bibitem {Cover}T. Cover and J. Thomas, \textit{Elements of Information
Theory}, John Wiley \& Sons, Inc., New York, 2006.

\bibitem {Carleial}A. B. Carleial, \textquotedblleft Interference
channels,\textquotedblright\ \textit{IEEE Trans. Inform. Theory}, vol. 24, no.
1, pp. 60--70, Jan. 1978.

\bibitem {Slepian Wolf}D. Slepian and J.K. Wolf, \textquotedblleft A coding
theorem for multiple access channels with correlated
sources,\textquotedblright\ \textit{Bell Systems Tech. J.}, vol. 52, pp.
1037-1076, Sept. 1973.

\bibitem {Liang Kramer}Y. Liang and G. Kramer, \textquotedblleft Rate regions
for relay broadcast channels,\textquotedblright\ \textit{IEEE\ Trans. Inform.
Theory}, vol. 53, no. 10, pp. 3517-3535, Oct. 2007.

\bibitem {Simeone Asilomar}O. Simeone, O. Somekh, G. Kramer, H. V. Poor and S.
Shamai (Shitz), \textquotedblleft Three-user Gaussian multiple access channel
with partially cooperating encoders,\textquotedblright\ to appear in
\textit{Proc. Asilomar Conference on Signals, Systems and Computers},
Monterey, CA, 2008.

\bibitem {Steiner}A. Steiner, A. Sanderovich and S. Shamai (Shitz),
\textquotedblleft The multi-session multi-layer broadcast approach for two
cooperating receivers,\textquotedblright\ in \textit{Proc. IEEE International
Symposium Inform. Theory} (ISIT 2008), Toronto, ON, Canada, July 6-11, 2008.

\bibitem {Avestimehr}R. Etkin, D. Tse and Hua Wang, \textquotedblleft Gaussian
interference channel capacity to within one bit,\textquotedblright\ submitted [arXiv:cs/0702045v1].

\bibitem {Avestimehr 1}G. Bresler and D. Tse, \textquotedblleft The two-user
Gaussian interference channel: a deterministic view,\textquotedblright%
\ \textit{European Trans. Telecommunications}, vol. 19, no. 4, pp. 333-354, 2008.
\end{thebibliography}
\end{document}